\documentclass[a4paper,english]{lipics-v2016}
\usepackage{microtype}%if unwanted, comment out or use option "draft"
\usepackage{mathtools}
\usepackage{booktabs}
\usepackage{bm}
\usepackage{tikz}
\usetikzlibrary{arrows,shapes,automata,backgrounds,calc,positioning,fit,petri,decorations.pathmorphing}
\usepackage{thmtools, thm-restate}
\theoremstyle{plain}
\usepackage[linesnumbered,ruled,vlined]{algorithm2e}

\graphicspath{{./graphics/}}%helpful if your graphic files are in another directory

\newcommand{\places}{\mathcal{P}}
\newcommand{\sysplaces}{\mathcal{P}_S}
\newcommand{\envplaces}{\mathcal{P}_E}
\newcommand{\transitions}{\mathcal{T}}
\newcommand{\net}{\mathcal{N}}
\newcommand{\cut}{\mathcal{C}}
\newcommand{\marking}{\mathcal{M}}
\newcommand{\game}{\mathcal{G}}
\newcommand{\reachablemarkings}[1]{\mathcal{R}(#1)}
\newcommand{\post}[1]{\mathit{post}(#1)}
\newcommand{\pre}[1]{\mathit{pre}(#1)}
\newcommand{\past}[1]{\mathit{past}(#1)}
\newcommand{\flow}{\mathcal{F}}
\newcommand{\flows}{\mathrel{\mathcal{F}}}
\newcommand{\badmarkings}{\mathcal{B}}
\newcommand{\conflict}{\mathrel{\sharp}}
\newcommand{\markingstep}[1]{\mathrel{| #1 \rangle}}
\newcommand{\initialmarking}{\mathit{In}}
\newcommand{\innet}[1]{\bm{{#1}}}
\newcommand{\optg}[1]{\textcolor{red}{#1}}

\title{Synthesis in Distributed Environments\footnote{Supported by the European Research Council (ERC) Grant OSARES (No. 683300).}}

\author[1]{Bernd Finkbeiner}
\author[2]{Paul G\"olz}
\affil[1]{Saarland University, Saarbr\"ucken, Germany\\
  \texttt{finkbeiner@react.uni-saarland.de}}
\affil[2]{Saarland University, Saarbr\"ucken, Germany\\
  \texttt{pgoelz@cs.cmu.edu}}
\authorrunning{Bernd Finkbeiner and Paul G\"olz}

\Copyright{Bernd Finkbeiner and Paul G\"olz}

\subjclass{F.3.1 Specifying and Verifying and Reasoning about Programs}\keywords{reactive synthesis, distributed information, causal memory, Petri nets}

\EventEditors{}
\EventNoEds{0}
\EventLongTitle{}
\EventShortTitle{}
\EventAcronym{}
\EventYear{}
\EventDate{}
\EventLocation{}
\EventLogo{}
\SeriesVolume{}
\ArticleNo{}

\begin{document}

\maketitle

\begin{abstract}
Most approaches to the synthesis of reactive systems study the problem in terms of a two-player game with complete observation.  
In many applications, however, the system's environment consists of several distinct entities, and the system must actively communicate with these entities in order to obtain information available in the environment.
In this paper, we model such environments as a team of players and keep track of the information known to each individual player.
This allows us to synthesize programs that interact with a distributed environment and leverage multiple interacting sources of information.

The synthesis problem in distributed environments corresponds to
solving a special class of Petri games, i.e., multi-player games
played over Petri nets, where the net has a distinguished token
representing the system and an arbitrary number of tokens representing
the environment. While, in general, even the decidability of Petri games
is an open question, we show that the synthesis problem in distributed
environments can be solved in polynomial time for nets with up to two
environment tokens. For an arbitrary but fixed number of three or
more environment tokens, the problem is NP-complete. If the number of
environment tokens grows with the size of the net, the problem is
EXPTIME-complete.
\end{abstract}

\section{Introduction}
Automating the creation of programs is one of the most ambitious goals
in computer science.  Given a specification, a synthesis algorithm
either generates a program that satisfies the specification or
determines that no such program exists.  The promise of synthesis is
to let programmers work on a more abstract level and thus to
fundamentally simplify the development of complex software.

Most current synthesis approaches (cf. \cite{conf/cav/JobstmannGWB07,conf/tacas/Ehlers11,conf/cav/BohyBFJR12,journals/corr/JacobsBBK0KKLNP16,DBLP:conf/tacas/FaymonvilleFRT17}) are based on the
game-theoretic approach, originally introduced by Büchi and Landweber
\cite{buchilandweber}, in which the synthesis problem is seen as a
two-player game with complete observation, played between a
\emph{system} player and an \emph{environment} player.  The goal of
the system player is to ensure that the specification is satisfied;
the goal of the environment player is to ensure a violation.  A
winning strategy for the system player defines a control program that
reads in the decisions of the environment as its inputs and produces
the decisions of the system as its outputs.

A fundamental limitation of the standard game-theoretic formulation is that the environment is a monolithic block. In many applications, however, the environment consists of several distinct entities, and the system must actively communicate with these entities in order to obtain information available in the environment. In this paper, we introduce the \emph{synthesis problem in distributed environments}.
As in the standard approach, we view the synthesis problem as a game between the system and the environment.
However, rather than considering the environment as a single \emph{player} in this game, we consider it as a \emph{team} consisting of several players that may carry different information.
Both the individual environment players and the system player can increase their knowledge by interacting with other players.

The problem is related to, but very different from, the \emph{distributed synthesis} problem~\cite{pnuelirosnerundec}.
In distributed synthesis, it is the \emph{system} that is partitioned into multiple players, corresponding to multiple processes.
The key difficulty here is to coordinate the strategies of the system players.
In the synthesis problem in distributed environments, it is instead the \emph{environment} that consists of multiple entities.
The key difficulty here is for the system player to synchronize with the right environment players at the right points in time.

We study the synthesis problem in distributed environments in the framework of \emph{Petri games}~\cite{petrigames}.
The players of a Petri game are represented as the tokens of a Petri net, partitioned into the system and environment players. Synthesis in distributed environments corresponds to Petri games with a single system token and multiple environment tokens. We assume that the underlying Petri net is bounded, i.e., only a bounded number of players can be generated over the course of a game.
For unbounded nets, Petri games are known to be undecidable~\cite{petrigames}.

The players of a Petri game advance asynchronously except for synchronous interactions, in which players exchange knowledge.
We assume that, whenever multiple players interact, they exchange information both truthfully and maximally.
This model of knowledge is called \emph{causal memory}.
In this paper, we restrict our synthesis to \emph{safety} specifications, i.e., the system must prevent the global state from entering certain bad configurations.

We illustrate our setting with a small access control example. Suppose you would like to synthesize a lock controller for a safe that contains sensitive business information.
Corporate policy mandates that the safe may only be jointly opened by two employees and that both must previously have confirmed their identity with a corresponding authentication authority.
The environment of the lock controller thus consists of four independent players: the employees $e^1$ and $e^2$ and their authenticators $a^1$ and $a^2$.
These entities interact with each other (when $a^1$ authenticates $e^1$ or $a^2$ authenticates $e^2$) and with the system player (when $e^1$ or $e^2$ request the safe to open). Since there is no direct interaction between the lock controller and the authenticators, the knowledge about the authentication must be provided to the lock controller by the employees.\footnote{In Petri games, all players are truthful. Think of the tokens as carriers of information, e.g., a cryptographically secured smart card carried by the employee.}

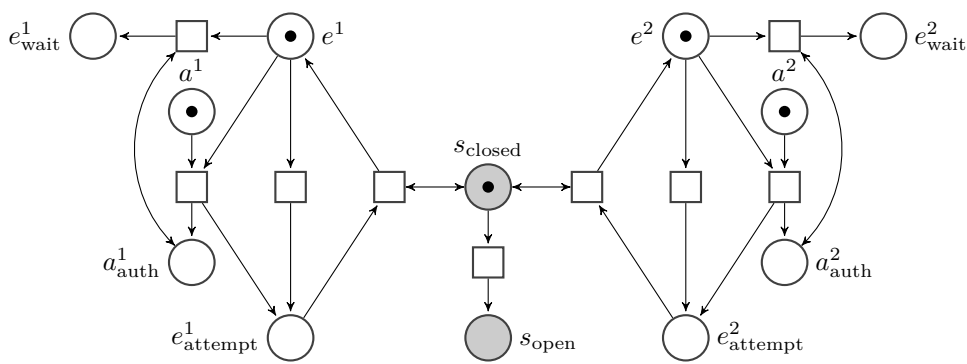
\begin{figure}
\center
\begin{tikzpicture}[on grid,node distance=1cm and 1.3cm,label distance=-.03cm,>=stealth',bend angle=45,auto]
        \tikzstyle{env}=[circle,thick,draw=black!75,minimum size=6mm]
        \tikzstyle{sys}=[env,fill=black!20]
        \tikzstyle{trans}=[rectangle,thick,draw=black!75,minimum size=4mm]

    \node[env] (e21) [tokens=1,label=left:$e^2\vphantom{e^2_\text{wait}}$] {};
		\node[trans] (te23) [right=of e21] {}
			edge [pre] (e21);
		\node[env] (e2p) [right=of te23,label=right:$e_\text{wait}^2$] {}
			edge [pre] (te23);
    \node[env] (e2a1) [tokens=1,below=of te23,label=above:$a^2$] {};
		\node[trans] (te2a) [below=of e2a1] {}
			edge [pre] (e2a1)
			edge [pre] (e21);
		\node[env] (e2a2) [below=of te2a,label=right:$a_\text{auth}^2$] {}
			edge [pre] (te2a)
			edge [<->,bend right] (te23);
		\node[trans] (te21) [left=of te2a] {}
			edge [pre] (e21);
		\node (phantom2) [below=of e2a2] {};
		\node[env] (e22) [left=of phantom2,label=right:$e_\text{attempt}^2$] {}
			edge [pre] (te21)
			edge [pre] (te2a);
		\node[trans] (te22) [left=of te21] {}
			edge [pre] (e22)
			edge [post] (e21);
			
		\node[sys] (s1) [tokens=1,left=of te22,label=above:$s_\text{closed}$] {}
			edge [<->] (te22);
		\node[trans] (ts) [below=of s1] {}
			edge [pre] (s1);
    \node[sys] (s2) [below=of ts,label=right:$s_\text{open}\vphantom{^1}$] {}
			edge [pre] (ts);
			
    \node[env] (e11) [tokens=1,label=right:$e^1\vphantom{e^1_\text{wait}}$] at ($(e21)+(-5.2cm,0)$) {};
		\node[trans] (te13) [left=of e11] {}
			edge [pre] (e11);
		\node[env] (e1p) [left=of te13,label=left:$e_\text{wait}^1$] {}
			edge [pre] (te13);
		\node[env] (e1a1) [tokens=1,below=of te13,label=above:$a^1$] {};
		\node[trans] (te1a) [below=of e1a1] {}
			edge [pre] (e1a1)
			edge [pre] (e11);
		\node[env] (e1a2) [below=of te1a,label=left:$a_\text{auth}^1$] {}
			edge [pre] (te1a)
			edge [<->,bend left] (te13);
		\node[trans] (te11) [right=of te1a] {}
			edge [pre] (e11);
		\node (phantom1) [below=of e1a2] {};
		\node[env] (e12) [right=of phantom1,label=left:$e_\text{attempt}^1$] {}
			edge [pre] (te11)
			edge [pre] (te1a);
		\node[trans] (te12) [right=of te11] {}
			edge [pre] (e12)
			edge [post] (e11)
			edge [<->] (s1);
\end{tikzpicture}
  \caption{Petri game of the access control example. If the system player lies in $s_\text{open}$ while there is still a player in $a^1$ or $a^2$, the system immediately loses the game.}
\label{fig:access}
\end{figure}
Figure~\ref{fig:access} shows how our access control scenario can be modeled as a Petri game.
Players are represented by tokens (dots) that move between places (circles) using transitions (squares).
The system player, who only moves between places marked in gray, starts in a place indicating that the safe is closed.
The game allows her to consult with any employee and remain in her position, or to move to the place $s_\text{open}$ to open the safe.
The first employee starts in $e^1$ and can either directly move to $e_\text{attempt}^1$ or can synchronize with her authenticator to move there.
In the latter case, the authenticator simultaneously moves to $a_\text{auth}^1$, where she cannot authenticate $e^1$ a second time.
When the employee is in $e_\text{attempt}^1$, the system player can choose to synchronize with her, moving the employee back to $e^1$ and exchanging all knowledge between the players.
In particular, the system player learns whether the employee was authenticated.
Afterwards, the employee can attempt to open the safe again, for example to make up for not being authenticated the last time.
If the employee has already authenticated, she can alternatively move to $e_\text{wait}^1$ and remain there.
This possibility forces the locking mechanism to stop waiting for communication once it knows enough and to unlock the safe instead.
The second employee is modeled symmetrically.
To prevent the system from unlocking prematurely, we declare that all situations in which the safe is open but in which one authenticator has not moved yet as losing for the system.

A winning strategy for this game, as found by our synthesis algorithm, would be to allow communication with $e^1$ and nothing else until (possibly never) the system learns that the employee has authenticated.
Then, it allows communication with $e^2$ until the same is true for the second employee.
Finally, it opens the safe.

\subparagraph*{Related work}
Synthesis in distributed environments is related to planning under partial observation \cite{aimodern} in that our strategies also combine information gathering and action.
However, the classical partial-information setting does not capture the knowledge of different actors.
With causal memory, a player's knowledge naturally refers to past observations and to the knowledge of other players.
Synthesis in distributed environments can be expressed as a control problem \cite{causalmemory,acyclicarchitectures} for Zielonka's asynchronous automata \cite{zielonka}.
Because this model is very expressive, all known decidability results assume strong restrictions on the communication architecture.
Since our environment players are allowed to freely interact with each other and with the system, we cannot apply these results.
Petri games were introduced in \cite{petrigames} and there is growing tool support for solving Petri games~\cite{adam,synt2017}. The decidability of general Petri games is an open question. The only previously known decision procedure is restricted to the case of a single environment token~\cite{petrigames}. In this paper, we solve the complementary case, where the number of environment tokens is unbounded (but there is only one system token). There is also a semi-algorithm for solving Petri games~\cite{bounded}. This approach finds finitely representable winning strategies, but does not terminate if no winning strategy exists.

\subparagraph*{Contributions}
Our main technical contribution is an EXPTIME algorithm for deciding bounded Petri games with one system player and an arbitrary number of environment players.
Previously, the synthesis problem for Petri games with more than one environment player was open.
We provide a matching lower bound to show that our algorithm is asymptotically optimal.
If the number of environment players is kept constant, we show that the problem can be solved in polynomial time for up to two environment players whereas it is NP-complete for three or more environment players.
The following table sums up the complexity of deciding $k$-bounded Petri games with one system player and $e$ environment players, for any $k \geq 1$:
\begin{center}
  \begin{tabular}{l l}
    \toprule
    $e \leq 2$ & P \\
    $e \geq 3$ & NP-complete \\
    $e$ grows with net & EXPTIME-complete \\
    \bottomrule
  \end{tabular}
\end{center}

\section{Petri nets}
\label{sec:petrinets}
We recall notions from the theory of Petri nets as used in \cite{petrigames}.
A tuple $\net = (\places, \transitions, \flow, \initialmarking)$ is called a \emph{Petri net} if it satisfies the following conditions:
\begin{itemize}
  \item The set of \emph{places} $\places$ and the set of
    \emph{transitions} $\transitions$ are disjoint;
  \item The \emph{flow relation} $\flow$ is a multiset over $(\places \times \transitions) \cup (\transitions \times \places)$, i.e., $\net$ is a directed, bipartite multigraph with nodes
    $\places \cup \transitions$ and edges given by $\flow$. We use the term \emph{nodes} to refer to places and transitions simultaneously.
    For nodes $x,y$, we write $x \flows y$ to denote $(x,y) \in \flow$;
  \item The \emph{initial marking} $\initialmarking$ is a finite multiset over $\places$;
  \item We require finite synchronization and nonempty pre- and postconditions:
    For a node $x$, define the \emph{precondition} as a multiset $\pre{x}$ such that $\pre{x}(y) = \flow(y,x)$ for all nodes $y$ and similarly define the \emph{postcondition} by $\post{x}(y) = \flow(x,y)$.
    Then, all transitions $t$ must satisfy $0 < |\pre{t}| < \infty$ and $0 < |\post{t}| < \infty$.
\end{itemize}
A net is called \emph{finite} if it contains finitely many nodes.

By convention, the components of a net $\net$ are named $\places$, $\transitions$, $\flow$ and $\initialmarking$, and similarly for nets named $\net_1$, $\net^\sigma$, $\net^U$, etc.
We graphically specify Petri nets as multigraphs, where places are represented by circles, transitions by squares and the flow relation by arrows.
In addition, the number of dots in a place reflects the multiplicity of this place in the initial marking.
Apart from the gray color of certain places, Fig.~\ref{fig:access} shows a Petri net with named places.

A \emph{marking} $\marking$ of $\net$ is a finite multiset over $\places$.
We think of the Petri net as a board on which a finite number of \emph{tokens} moves between places by using transitions.
A marking then represents a certain configuration by listing the current number of tokens on every place.
We can move from one marking to another by firing a transition $t$, i.e., by removing tokens in $\pre{t}$ and putting tokens into $\post{t}$ instead.
If the total number of tokens changes in this process, we think of such transitions as generating or consuming tokens.
We say that $t$ is \emph{enabled} in a marking $\marking$ if $\pre{t} \subseteq \marking$.
If this is the case, we can obtain a new marking $\marking' \coloneqq \marking - \pre{t} + \post{t}$ by \emph{firing} $t$, and we write $\marking \markingstep{t} \marking'$ to denote that $\marking'$ can be constructed from $\marking$ and $t$ in this way.
A marking is said to be \emph{reachable} if it can be reached from the initial marking by firing a finite sequence of transitions.
We generalize preconditions and postconditions to sets $S$ of nodes by defining $\pre{S} \coloneqq \biguplus_{x \in S} \pre{x}$ and analogously for $\post{S}$.
A Petri net is \emph{$k$-bounded} for a natural number $k \geq 1$ if, for all reachable markings $\marking$ and places $p$, $\marking(p) \leq k$ holds.
We call a net \emph{bounded} if it is $k$-bounded for some $k$.

We are mainly interested in Petri nets as a model for the causal dependencies between events.
These dependencies are made explicit in \emph{occurrence nets}, certain acyclic nets in which each place has a unique causal history.
Before giving their definition, we introduce notation to capture different kinds of causal relationships between nodes.
We denote the transitive closure of the support of $\flow$ by $<$ and its reflexive and transitive closure by $\leq$.
We call $x$ and $y$ \emph{causally related} if $x \leq y$ or $y \leq x$.
The \emph{causal past} of a node $x$ is the set $\past{x} \coloneqq \{y \in \places \cup \transitions \mid y \leq x\}$.
We extend this notion to sets $S$ of nodes by setting $\past{S} \coloneqq \bigcup_{x \in S} \past{x}$.
Apart from being causally related, two nodes $x, y$ might also be mutually exclusive, i.e., they might be the result of alternative, nondeterministic choices.
We say that $x$ and $y$ are \emph{in conflict}, for short $x \conflict y$,
if there exists a place $p$, $p \neq x, p \neq y$, such that $x$ and $y$ can be reached following the flow relation from $p$ via different outgoing transitions.
If $x$ and $y$ are neither causally related nor in conflict, we call them \emph{concurrent}.

An \emph{occurrence net} is a net $\net$ that satisfies all of the following conditions:
the pre- and postconditions of transitions are sets, not general multisets;
each place has at most one incoming transition;
the initial marking is the set $\{p \in \places \mid \pre{p} = \emptyset\}$;
the inverse flow relation $\flow^{-1}$ is well-founded, i.e., if we start from any node and follow the flow relation backwards, we eventually reach a place in the initial marking;
no transition is in conflict with itself.
Occurrence nets are $1$-bounded, i.e., their reachable markings are sets.

We call a maximal set $\cut$ of pairwise concurrent places in an occurrence net a \emph{cut}.
In Appendix~\ref{sec:cutsmarkings}, we prove that the \emph{finite} cuts of an occurrence net are exactly its reachable markings.
We further prove that the occurrence nets that we will work with only have finite cuts.
Thus, for our purposes, we can use the terms interchangeably (Corollary~\ref{lem:boundedbranchingcutsmarkings}).

A \emph{homomorphism} from a Petri net $\net_1$ to a Petri net $\net_2$ is a function $\lambda : \places_1 \cup \transitions_1 \rightarrow \places_2 \cup \transitions_2$ that only maps places to places and transitions to transitions such that, for all $t \in \transitions_1$, $\lambda[\pre{t}] = \pre{\lambda(t)}$ and $\lambda[\post{t}] = \post{\lambda(t)}$.
$\lambda$ is called \emph{initial} if additionally $\lambda[\initialmarking_1] = \initialmarking_2$ holds.

An \emph{initial branching process} $\beta$ of a net $\net$ is a pair $(\net^U, \lambda)$ where $\net^U$ is an occurrence net and $\lambda$ is an initial homomorphism from $\net^U$ to $\net$ such that
$\forall t_1, t_2 \in \transitions^U.\;(\pre{t_1} = \pre{t_2} \land \lambda(t_1) = \lambda(t_2)) \rightarrow t_1 = t_2$.
Conceptually, a branching process describes a subset of the possible behavior of a net as an occurrence net.
If a place or a transition in the original net can be reached on different paths or with different knowledge, the branching process splits up this node.
The homomorphism $\lambda$ is used to label those multiple instances with the original node in $\net$.
The additional condition means that the branching process may not split up a transition unnecessarily:
For the same precondition, at most one instance of a certain transition can be present in the branching process.

\section{Petri games}
\label{sec:petrigames}
In a Petri game, we partition the places of a finite Petri net into two disjoint subsets: the \emph{system places} $\sysplaces$ (represented in gray) and the \emph{environment places} $\envplaces$ (represented in white).
For convenience, we write $\places$ for the set of all places of the game $\sysplaces \cup \envplaces$.
A token on a system place represents a \emph{system player}, a token on an environment place an \emph{environment player}.
Additionally, a Petri game also identifies a set of \emph{bad markings} $\badmarkings$, which the system players need to avoid.\footnote{This is more general than in \cite{petrigames}, where instead of avoiding a set of arbitrary markings, the system tries to avoid all markings that have a nonempty intersection with a set of bad places.
\cite{bounded} also uses arbitrary sets of bad markings.
Since the hardness proofs in Theorems~\ref{thm:exphardness} and \ref{thm:nphardness} only use bad markings of this shape, this generalization does not increase the computational hardness of our setting.
In our complexity analyses in Theorems~\ref{thm:exp}, \ref{thm:nphardness}, \ref{thm:polynomialtwo} and in Appendix~\ref{sec:algocomp}, we do not commit to a specific input encoding of bad markings such that our results remain valid if a set of bad places is given instead.}
If the game reaches a marking $\marking$ in $\badmarkings$, the environment wins; the system wins if this is never the case.
Formally, a \emph{Petri game} $\game$ is a tuple $(\sysplaces, \envplaces, \transitions, \flow, \initialmarking, \badmarkings)$.
We call $\net^\game \coloneqq (\places, \transitions, \flow, \initialmarking)$ the \emph{underlying net} of $\game$.

Transitions whose entire precondition belongs to the environment are called \emph{purely environmental}.
Otherwise, we call the transition a \emph{system transition}.

Since Petri games aim to model the information flow in a system, a system player's decisions may only depend on information that she has witnessed herself or that she has obtained by communicating with other players.
We thus describe strategies of the system as branching processes of the underlying net of the game, where the causal dependencies are made explicit.
While the game is played on the underlying net, the strategy keeps track of the current state of the game as well as its causal history.
As we show in Lemma~\ref{lem:branchingmarkings} in Appendix~\ref{sec:cutsmarkings}, every reachable marking of the branching process corresponds to a reachable marking in the underlying net.
The marking in the strategy might have less enabled transitions than the one in the underlying net, which means that the strategy can prevent certain transitions from firing.
The game progresses by nondeterministically firing transitions that are allowed by the strategy.
No matter which transitions are fired in which order, the system players need to ensure certain properties of the game.
Because of this, it is sometimes useful to think of these choices as being made by an adversarial scheduler.

A winning, deadlock-avoiding \emph{strategy} is an initial branching process $\beta_\sigma = (\net^\sigma, \lambda)$ of the underlying net of the game that satisfies the following four conditions:
\begin{description}
\item[justified refusal] Let $S$ be a set of pairwise concurrent places in $\places^\sigma$ and $\innet{t}$ be a transition \emph{in the underlying net}, where
  $\lambda[S] = \pre{\innet{t}}$ but there is no $t \in \transitions^\sigma$ such that $\lambda(t) = \innet{t}$ and $\pre{t} = S$.
    Then, there must be a place $s \in S \cap \sysplaces^\sigma$ such that $\innet{t} \notin \lambda(\post{s})$.
\item[safety] For all $\marking \in \reachablemarkings{\net^\sigma}$, $\lambda[\marking] \notin \badmarkings$.
\item[determinism] For all $s \in \sysplaces^\sigma$ and all reachable markings $\marking$ in $\net^\sigma$ that contain $s$, there is at most one transition $t \in \post{s}$ that is enabled in $\marking$.
\item[deadlock avoidance] For all $\marking \in \reachablemarkings{\net^\sigma}$ we require that, if any transition of the underlying net is enabled in $\lambda[\marking]$, then some transition in the strategy must be enabled in $\marking$.
\end{description}
In the above conditions, we extended the notion of system places to the strategy by setting $\sysplaces^\sigma \coloneqq \places^\sigma \cap \lambda^{-1}(\sysplaces)$. We similarly define the environment places of the strategy as $\envplaces^\sigma \coloneqq \places^\sigma \cap \lambda^{-1}(\envplaces)$.
To distinguish more clearly between nodes in the strategy and nodes in the underlying net, we always use bold variable names such as $\innet{p}$ or $\innet{t}$ for the latter.

Justified refusal means that a system player influences the course of the game by refusing to take part in certain transitions in her postcondition.
Even if every place in $\pre{\innet{t}}$ contains a token for some $\innet{t} \in \transitions$, the transition can fire iff, for every place in $\pre{\innet{t}} \cap \sysplaces$, the corresponding system player allows this transition.
In particular, purely environmental transitions cannot be restricted by the strategy.
More precisely, the condition refers to all possible preconditions $S$ where a transition could have been added to the strategy, but was not.
If no instance $t$ of $\innet{t}$ with the right precondition exists so far, there must be a system place in $S$ that refuses to take part in any instance of $\innet{t}$.
Note that a system player can only refuse all transitions in the strategy with the label $\innet{t}$ or must allow all of them.

The safety objective requires that the game never reaches a bad marking.
Determinism enforces that, from a system player's perspective, all sources of uncertainty are in the vicinity of an environment player.
This does not prevent a system player from allowing multiple transitions, as long as these transitions are enabled in different markings.

Finally, we require the strategy to avoid deadlocks.
Without this condition, a strategy might simply refuse to fire any system transition at all.
In general, the system prefers to fire less transitions since they might potentially lead to bad markings and since allowing too many of them might cause nondeterminism.
The criterion enforces that, whenever no purely environmental transition is enabled in a marking but some system transition is enabled, the strategy must allow one of them in order to keep the game going.
This still allows the strategy to enter markings in which no transition is enabled at all.
Similarly, a system player may refuse all transitions in her postcondition as long as she knows that the game will always allow another player to move.

\section{Reduction to games over finite graphs}
We wish to decide whether a $k$-bounded Petri game with one system player admits a winning, deadlock-avoiding strategy.
In case of a positive answer, we also want to obtain a description of such a strategy.
Note that the system player's decisions can be based on an unboundedly growing amount of information.
Because of this, it is not at all obvious that the existence of a strategy is decidable and that strategies can be represented in finite space.

In this section and the next, we show that the decision problem is EXPTIME-complete in the size of the net.
We establish the upper bound through a many-one reduction to a complete-observation game over a finite graph.
We consider Petri games with a single system player, i.e., 
all reachable markings $\marking$ contain exactly one system place, which we denote by $\innet{s_\marking}$.
In the cuts $\cut$ of a strategy, we denote the unique system place by $s_\cut$.

\begin{figure}
  \begin{align*}
  \mathcal{V}_0\optg{'} = &\left\{ \left( \marking, \top \optg{, \{ \innet{s_\marking} \}} \right) \mid \marking \in \reachablemarkings{\net} \right\} \\
  \mathcal{V}_1\optg{'} = &\left\{ \left( \marking, c\optg{, R} \right) \mid \marking \in \reachablemarkings{\net}; c \subseteq \post{\innet{s_\marking}} \optg{; \innet{s_\marking} \in R \subseteq \marking} \right\} \\
  \mathcal{I}\optg{'} =&\,(\initialmarking, \top \optg{, \{\innet{s_\initialmarking}\}}) \\
  \mathcal{E}\optg{'} = &\left\{ \left( \marking, \top \optg{, \{\innet{s_\marking}\}} \right) \rightarrow \left( \marking, c \optg{, \{\innet{s_\marking}\}} \right) \mid c \subseteq \post{\innet{s_\marking}} \right\} \tag{E\optg{'}1} \\
       \cup &\left\{ \left( \marking, c\optg{, R} \right) \rightarrow \left( \marking', c\optg{, R'} \right) \,\middle|\, \begin{aligned}
         &\innet{t}\text{ purely environmental transition}; \marking \markingstep{\innet{t}} \marking'\optg{;} \\
                                                          &\optg{\innet{o} \in \post{\innet{t}}; R' = R - \pre{\innet{t}} + \left\{ \innet{o} \right\}}
       \end{aligned} \right\} \tag{E\optg{'}2} \\
       \cup &\left\{ \left( \marking, c\optg{, R} \right) \rightarrow \left( \marking', \top \optg{, \{ \innet{s_{\marking'}} \}} \right) \,\middle|\,
         \begin{aligned}
           &\text{$\innet{t}$ system transition}; \innet{t} \in c; \marking \markingstep{\innet{t}} \marking'\optg{;}\\
           &\optg{R \subseteq \pre{\innet{t}}}
         \end{aligned}
       \right\} \tag{E\optg{'}3}\\
  \mathcal{X}\optg{'} =& \left\{ (\marking, c\optg{, R}) \mid \marking \in \badmarkings \right\} \tag{X\optg{'}1} \\
      \cup & \left\{ (\marking, c\optg{, R}) \mid \innet{t}, \innet{t'} \in c; \innet{t} \neq \innet{t'}; \text{both enabled in $\marking$} \right\} \tag{X\optg{'}2a} \\
      \cup & \left\{ (\marking, c\optg{, R}) \mid \innet{t} \in c; \text{enabled in $\marking$}; 0 < \pre{\innet{t}}(\innet{p}) < \marking(\innet{p}) \text{ for some $\innet{p} \in \places$} \right\} \tag{X\optg{'}2b} \\
      \cup & \left\{ (\marking, c\optg{, R}) \mid \text{Some $\innet{t} \in \transitions$ enabled$;$ all such $\innet{t}$ involve the system and $\innet{t} \notin c$} \right\} \tag{X\optg{'}3}
\end{align*}%
  \caption{Description of the two graph games constructed from $\game$. For the components of $\mathit{Graph}(\game)$, ignore all colored parts. Including them, we get the components of $\mathit{Graph}'(\game)$.}
\label{fig:ggame}
\end{figure}

For a given Petri net $\game$ with underlying net $\net$, Fig.~\ref{fig:ggame} defines the components of the translated graph game $\mathit{Graph}(\game) = (\mathcal{V}_0, \mathcal{V}_1, \mathcal{I}, \mathcal{E}, \mathcal{X})$ if we ignore all colored parts.
The set of \emph{vertices} $\mathcal{V}$ consists of two disjoint subsets $\mathcal{V}_0$ and $\mathcal{V}_1$, which describe the vertices belonging to players~0 and 1, respectively.
The game begins in the initial vertex $\mathcal{I}$.
From a vertex $v \in \mathcal{V}$, the current player chooses an outgoing \emph{edge} in $\mathcal{E}$.
A \emph{play}, i.e., a maximal sequence $\mathcal{I} = v_0\,v_1\,\dots$ of vertices with $(v_i, v_{i+1}) \in \mathcal{E}$ for all $i$, is winning for Player~0 if no vertex is an element of the bad vertices $\mathcal{X}$.
A \emph{strategy} $T_\sigma$ (for Player~0) is a $\mathcal{V}$-labeled tree whose root is labeled with $\mathcal{I}$.
If a node is labeled with a vertex in $\mathcal{V}_1$, its children are labeled with all successor vertices.
Otherwise, it has a single child labeled with one particular successor.
The strategy is \emph{winning} if all maximal paths through it are labeled with winning plays.
All such games are \emph{memoryless determined}:
If there is any winning strategy, there exists a winning strategy that selects, from any two nodes with the same label, the same successor vertex.

The vertices of the game essentially represent the reachable markings of the Petri game and Player~1 moves between markings by firing enabled transitions.
This means that Player~1 plays the role of both the environment and the nondeterminism stemming from different schedulings.
Player~0, who represents the system, can only act by refusing to allow some transitions in the postcondition of the single system place in the marking.
Since these decisions should not depend on scheduled, purely environmental transitions that the system would not yet know in the Petri game, Player~0 is forced to choose directly after the system player has taken a transition.
Similarly to \cite{petrigames}, we therefore add a \emph{commitment}, i.e., a set $c \subseteq \post{\innet{s_\marking}}$, to each vertex of the graph game.
The commitment keeps track of the set of outgoing transitions of the current system place that the system player allows.
Player~0's vertices are marked with $\top$ instead of a commitment to denote that she needs to decide on a commitment in the next step (E1).
Player~1's choices are then restricted such that she can fire all purely environmental transitions (E2) but can only fire system transitions that appear in the commitment (E3).
The bad vertices correspond to bad markings (X1), nondeterminism (X2a,\,X2b) and deadlock (X3).

To prove the reduction correct, we need to show that $\game$ has a winning, deadlock-avoiding strategy iff Player~0 has a winning strategy in $\mathit{Graph}(\game)$.
For this, we give translations between these types of strategies.

\subsection{From Petri game strategies to graph game strategies}
Assume that we are given a winning, deadlock-avoiding strategy $\beta_\sigma = (\net^\sigma,\lambda)$ for $\game$.
We inductively build a strategy $T_\sigma$ for $\mathit{Graph}(\game)$.
Whenever we encounter a node labeled with a vertex belonging to Player~0, we choose an outgoing edge, i.e., a suitable commitment.

For any such node, we look at the sequence of labels on the path that leads to it from the root.
This sequence is a prefix of a play, which we denote by $v_0 \, v_1 \, \dots \, v_r = (\marking, \top)$.
Edges of type (E1) in this prefix do not change the marking.
All other edges are associated with firing a transition.
Starting from the initial cut, we fire $\lambda$-preimages of these transitions one after another.
If multiple transitions could be responsible for the edge or if multiple preimages are enabled, choose one canonically.
For edges of type (E2), such preimages always exist because justified refusal does not allow $\beta_\sigma$ to restrict purely environmental transitions.
In the case of edges of type (E3), we make sure to only include transitions in the commitment if the existence of such preimages is ensured.
By consecutively firing such a sequence of transitions, we reach a cut $\cut$ such that $\lambda[\cut] = \marking$.
Set $c \coloneqq \{ \lambda(t) \mid t \in \post{s_\cut} \}$ and choose the outgoing edge leading to $(\marking, c)$ to construct the strategy.

For well-definedness, it remains to show that, when Player~1 schedules a system transition $\innet{t} \in c$ the next time, a preimage of this transition will be enabled in the cut $\cut'$ that corresponds to the node in the strategy.
Since, in between, only purely environmental transitions will be fired, $s_\cut$ will still be part of $\cut'$.
The system place has a preimage of $\innet{t}$ in its postcondition by the definition of $c$.
Therefore, a preimage enabled in $\cut'$ exists by justified refusal.

\begin{restatable}{theorem}{petrigraph}
  \label{thm:petrigraph}
  $T_\sigma$ is a winning strategy for Player~0.
\end{restatable}
\begin{proof}[Proof sketch (detailed in Appendix~\ref{sec:petrigraphproof})]
  Consider a node $n$ in $T_\sigma$ with the label $(\marking, c)$.
  As in the construction of the graph game strategy, we canonically fire transitions corresponding to the prefix until we reach a cut $\cut$ such that $\lambda[\cut] = \marking$.
  Now assume that $n$ is a bad vertex.
  Each kind of bad vertices (X1), (X2a), (X2b) or (X3) translates to a violation of the properties of a winning, deadlock-avoiding strategy in $\cut$, contradiction.
  Thus, no node is labeled with a bad vertex and the strategy is winning.
\end{proof}

\subsection{From graph game strategies to Petri game strategies}
\label{sec:graphpetri}
The converse direction is harder to prove.
So far, we have shown that, if the system can win a Petri game with incomplete information, Player~0 can also win a game with full information on the marking graph.
This is not surprising.
In this step however, we must show that this additional information does not give an advantage to Player~0 that the system does not have.
In the construction of $\mathit{Graph}(\game)$, we have already introduced commitments, which prevent Player~0 from using information about the scheduling of purely environmental transitions for her subsequent move.
However, Player~0 might still use this information to make her move after the next.
If the system player does not learn about the environment transition in her next step, this is an illegal flow of information.

The main idea now is that, while some parts of the graph game strategy do not correspond to a valid information flow in the Petri game, others do.
In these latter parts, the strategy contains all necessary decisions to win the Petri game.
Conceptually, we need to cut away unreasonable plays from the strategy.
Alternatively, we might say that a forbidden information flow only happens if Player~1 does not play in an intelligent way.
From Player~1's point of view, it is dangerous and unnecessary to schedule a purely environmental transition and then schedule a system transition unless the former is needed to enable the latter.
If she does so, Player~0 gains potentially useful information, which Player~1 could easily prevent by scheduling the purely environmental transition at a later point, i.e., when it is necessary to enable the next system transition or when a winning situation for Player~1 (bad marking, nondeterminism or deadlock) can be reached without any more moves by Player~0.
To make this idea formal, we construct another graph game $\mathit{Graph'}(\game)$, which restricts Player~1's moves to enforce the behavior described above.
Then, we can easily show that any winning strategy for $\mathit{Graph}(\game)$ translates to a winning strategy for $\mathit{Graph'}(\game)$, where Player~1 has fewer options.
In a second step, we will translate the strategy from $\mathit{Graph'}(\game)$ back to a strategy for the Petri game, which will prove the desired equivalence.

The new graph game $\mathit{Graph'}(\game) = (\mathcal{V}'_0, \mathcal{V}'_1, \mathcal{I}', \mathcal{E}', \mathcal{X}')$ is defined in Fig.~\ref{fig:ggame} by taking into account the colored parts.
The vertices of $\mathit{Graph}(\game)$ are extended by a third component, a \emph{responsibility multiset} $R$ over $\places$.
This multiset $R \subseteq \marking$ tracks the information generated by firing transitions.
At any point in the Petri game, a subset $S$ of the cut such that $\lambda[S] = R$ together carries the information about all fired transitions.
This notion is made precise in Lemma~\ref{lem:responsibility} in Appendix~\ref{sec:associatedlkcproof}.
After a transition has been fired, every token in its postcondition carries the information about the causal pasts of all participating tokens and about the fired transition itself.
For this reason, when an edge of type (E'2) fires a purely environmental transition $\innet{t}$, the tokens in $\pre{\innet{t}}$ are subtracted from $R$, and Player~1 chooses an arbitrary token $\innet{o} \in \post{\innet{t}}$, which will carry the information to the system player.
Edges of type (E'3) deal with $R$ similarly in that they also subtract the precondition from $R$ and instead add one element of the postcondition, namely the system place.
In contrast to $\mathit{Graph}(\game)$, these edges only allow system transitions if the responsibility multiset is included in the precondition, i.e., if the system player would directly learn about all previously scheduled transitions by taking this system transition.

\begin{restatable}{theorem}{graphgraphprime}
  \label{thm:graphgraphprime}
  If there is a winning strategy for $\mathit{Graph}(\game)$, there exists a winning strategy for $\mathit{Graph'}(\game)$.
\end{restatable}
\begin{proof}[Proof sketch (detailed in Appendix~\ref{sec:graphgraphprimeproof})] $\mathit{Graph'}(\game)$ only reduces Player~1's options.
\end{proof}

We now translate a winning strategy $T_\sigma$ for $\mathit{Graph'}(\game)$ back into a winning, deadlock-avoiding strategy for the Petri game.
Without loss of generality, we assume $T_\sigma$ to be memoryless.
We traverse the strategy tree in breadth-first order and inductively build the Petri game strategy $\beta_\sigma = (\net^\sigma, \lambda)$.
Simultaneously, we map each node of the tree to a nonempty set of cuts.
We call these cuts the \emph{associated cuts} of the node.
These cuts can be reached from $\initialmarking^\sigma$ by firing $\lambda$-preimages of transitions corresponding to the edges of types (E'2) and (E'3) on the path from the root to this node.
In particular, every such cut $\cut$ will satisfy $\lambda[\cut] = \marking$, where $\marking$ is the marking found in the label of the node.

We begin by mapping the root of the tree to a single cut $\initialmarking^\sigma$, i.e., a fresh set of places such that $\lambda[\initialmarking^\sigma] = \initialmarking$.
Then, we traverse $T_\sigma$ and distinguish between the different kinds of edges in the graph game by which the vertex of the currently visited node has been reached from its predecessor.
\begin{itemize}
  \item (E'1): Do not modify $\beta_\sigma$ and map the new node to the same cuts as its parent.
  \item (E'2) or (E'3): Let $\cut$ be one of the cuts associated with the parent node.
    Let $\innet{t}$ be a transition that could have been used in the definition of (E'2) or (E'3) to justify the existence of the edge.
    Finally, let $B$ be any subset of $\cut$ with $\lambda [B] = \pre{\innet{t}}$. Such a subset always exists because $\innet{t}$ is enabled in $\lambda[\cut]$.
    If it already exists, let $t \in \transitions^{\sigma}$ be a transition with $\pre{t}=B$ and $\lambda(t)=\innet{t}$.
    Else, create a new such transition and a fresh set of places as its postcondition such that $\lambda[\post{t}] = \post{\innet{t}}$.
    Choose $\cut'$ such that $\cut \markingstep{t} \cut'$.
    We map the new node to all cuts $\cut'$ that can be constructed from suitable $\cut$, $\innet{t}$ and $B$ in this way.
\end{itemize}

We need to show that $\beta_\sigma$ is a strategy.
First, we can easily see that the construction ensures all requirements of an occurrence net.
Furthermore, $\beta_\sigma$ is an initial branching process because $\lambda$ is an initial homomorphism and because we only add a new transition if no other transition with the same label and precondition exists.

Before we can prove that $\beta_\sigma$ satisfies the four axioms of a winning, deadlock-avoiding strategy, we need to show that the responsibility multiset construction works as intended.
First, we show that the construction prevents illegal information flows.
Whenever the system player moves in the graph game, she directly learns about all previously scheduled transitions.
Formally, nodes labeled with player-0 vertices are only mapped to cuts $\cut$ that are the \emph{last known cuts} of their respective system place $s_\cut$.
The last known cut of a place $x \in \places^\sigma$ is defined as $\mathit{LKC}(x) \coloneqq \{p \in \places^\sigma \mid p \nless x \land \forall t \in \pre{p}.\;t < x\}$.
In the terminology of \cite{esparza94}, this cut is the \emph{mapping cut} of $\past{x} \cap \transitions$, i.e., the cut reached by firing all transitions in the past of $x$.
The last known cut of $x$ has the special property that, for every cut $\cut$ with $x \in \cut$, the last known cut of $x$ lies in $\past{\cut}$ (Lemma~\ref{lem:lkcpast} in Appendix~\ref{sec:associatedlkcproof}).
\begin{restatable}{lem}{associatedlkc}
  \label{lem:associatedlkc}
  Let a node in $T_\sigma$ be labeled with a vertex belonging to Player~0 and let $\cut$ be one of its associated cuts. Then, $\cut = \text{LKC}(s_\cut)$.
\end{restatable}
\begin{proof}[Proof in Appendix~\ref{sec:associatedlkcproof}]
\end{proof}

Second, we need to show that the responsibility multiset construction does not overly restrict the scheduling.
For certain schedulings of purely environmental transitions, the responsibility multiset prevents a system transition from being fired even though it is enabled and in the commitment.
If, since the Player~0's last move, Player~1 had skipped firing all transitions that do not help to enable this system transition, the transition could be fired.
Therefore, the Petri game strategy contains all system transitions wherever they are not refused.
This is formally stated and proved in Lemma~\ref{lemselect} in Appendix~\ref{sec:selectproof}.

\begin{lemma}[safety]
  Let $\cut$ be a cut in $\net^\sigma$. Then, $\lambda[\cut] \notin \badmarkings$.
  \label{lem:safety}
\end{lemma}
\begin{proof}
  Consider the node $n$ for which $s_\cut$ was inserted into the strategy.
  This node must be labeled with a $\mathcal{V}_0$ vertex and must have $\mathit{LKC}(s_\cut)$ as one of its associated cuts by Lemma~\ref{lem:associatedlkc}.
  Since $\mathit{LKC}(s_\cut) \subseteq \past{\cut}$, there is a sequence of purely environmental transitions leading from $\mathit{LKC}(s_\cut)$ to $\cut$, by Lemma~\ref{lem:cutsreachable} in Appendix~\ref{sec:cutsmarkings}.
  Thus, from $n$'s unique successor, we can follow a corresponding sequence of type-(E'2) edges to a node $n'$ with $\cut$ as one of its associated cuts.
  If $\lambda[\cut]$ were a bad marking, $n'$ would be labeled with a bad vertex of type (X'1).
  Since $T_\sigma$ is a winning strategy, this is not the case.
\end{proof}

For the proofs of justified refusal (Lemma~\ref{lem:just}), determinism (Lemma~\ref{lem:determinism}) and deadlock avoidance (Lemma~\ref{lem:deadlock}), we refer the reader to Appendix~\ref{sec:graphprimepetriproof}.
As an immediate consequence, $\beta_\sigma$ is a winning, deadlock-avoiding strategy, which concludes the claimed equivalence:

\begin{theorem}
  \label{thm:graphprimepetri}
  If $\mathit{Graph'}(\game)$ has a winning strategy for Player~0, there exists a winning, deadlock-avoiding strategy for $\game$.
\end{theorem}

\section{Synthesis in distributed environments is EXPTIME-complete}
\begin{theorem}
\label{thm:exp}
For fixed $k \geq 1$, $k$-bounded Petri games with one system player and an arbitrary number of environment players can be decided in exponential time.
\end{theorem}
\begin{proof}
Our reduction allows to decide such Petri games $\game$ in exponential time:
The number of vertices in $\mathit{Graph}(\game)$ is bounded by $k^{|\places|} \cdot (2^{|\transitions|} + 1)$ and its local structure can be computed efficiently.
Since graph games with such safety winning conditions can be solved in linear time in the size of the game \cite[pp.~78--79]{apt2011lectures}, this requires exponential time in the size of the Petri game.

In Appendix~\ref{sec:algodec}, we describe an algorithm that evaluates the commitments symbolically and uses a SAT solver to speed up solving the game in practice.
If we solve the SAT instances through a na\"ive enumeration, we have an explicit EXPTIME algorithm, whose complexity is analyzed in Appendix~\ref{sec:algocomp}.
\end{proof}
\begin{restatable}{theorem}{exphardness}
\label{thm:exphardness}
Deciding $k$-bounded Petri games with one system player and an arbitrary number of environment players is EXPTIME-hard for any $k \geq 1$.
\end{restatable}
\begin{proof}[Proof sketch (detailed in Appendix~\ref{sec:exphardnessproof})]
We show hardness through a reduction from the EXPTIME-complete combinatorial game $G_5$ from \cite{combinatorialgames}.
This reduction is similar to the one given in \cite{petrigames} for the fragment with one environment player.
In $G_5$, two players, $P_S$ and $P_E$, take turns in switching the truth values of a finite set of boolean variables, one at a time.
Alternatively, they are allowed to pass.
The players operate on disjoint subsets of the variables.
Initially, the variables have predefined values.
If, at a certain point, a formula $\phi$ over the variables becomes satisfied, $P_E$ wins; else, $P_S$ wins.

For an instance of this game, we build a Petri game such that there is a winning, deadlock-avoiding strategy iff $P_S$ has a winning strategy in the original game.
Without loss of generality, let $\phi$ be given in negation normal form.
An example for the reduction is illustrated in Fig.~\ref{fig:1vnhardness} in Appendix~\ref{sec:exphardnessproof}.
Each variable is represented by an environment token moving between two places, indicating the variable's truth value.
An additional environment token keeps track of the current turn.
If it is $P_E$'s turn, this token synchronizes with one of the environment variables and switches its position.
If it is $P_S$'s turn, the token first informs the single system token of the previous moves and then enables the transitions for switching a system variable, from which the system token chooses one.

Instead of letting a player move, the turn token can permanently freeze the variables and prove that $\phi$ is satisfied.
For this, we have an additional environment token for every subformula, each with two places.
The turn token can move these tokens to their second place to prove that the subformula is satisfied.
For literals, the turn token needs to synchronize with the respective variable in the correct place.
For disjunctions, it must synchronize with the token of one of the subformulas, which must have been proved before.
For conjunctions, synchronization with both subformula tokens is required.
The bad markings are exactly those in which the entire formula $\phi$ is proved.
This game is $1$-bounded, thus $k$-bounded.
\end{proof}

\section{Sparse Petri games}
\label{sec:sparse}
The nets produced by our EXPTIME-hardness reduction contain a high number of environment tokens.
Because of this, the number of reachable markings grows exponentially and computational cost with it.
To study other sources of algorithmic hardness, we analyze the complexity of the problem for a fixed maximum number $p$ of environment players.
Then, we can bound the number of reachable markings by the polynomial $(|\places| + 1)^{(p + 1)}$ instead of by $(k + 1)^{|\places|}$.
For a fixed $p$, the problem is in NP:
We nondeterministically guess a commitment for every $\mathcal{V}_0$ vertex and verify in polynomial time that no bad vertices are reachable.%
\begin{restatable}{theorem}{nphardness}
  \label{thm:nphardness}
  For a fixed $p \geq 3$, deciding Petri games with one system player and $p$ environment players is NP-complete.
\end{restatable}
\begin{proof}[Proof sketch (detailed in Appendix~\ref{sec:nphardnessproof})]
  The upper bound has already been established.
  Show the lower bound by a reduction from the boolean satisfiability problem with 3-clauses (3SAT).
  For a given instance, construct a Petri game with three environment players and a single system player.
  For every clause, the single system player must allow at least one transition corresponding to a satisfied literal in the clause.
  Deadlock avoidance forces the system player to allow at least one such transition per clause.
  Nondeterminism prevents the system player from allowing two transitions corresponding to complementary literals.
\end{proof}

\begin{restatable}{theorem}{polynomialtwo}
  \label{thm:polynomialtwo}
  Petri games with one system player and at most two environment players can be decided in polynomial time.
\end{restatable}
\begin{proof}[Proof sketch (detailed in Appendix~\ref{sec:polynomialtwoproof})]
  We adapt the algorithm in Appendix~\ref{sec:algodec}, which evaluates commitments symbolically with a SAT solver.
  Due to the special structure of the SAT instances generated, we can add pre- and postprocessing steps such that the SAT queries only contain 2-clauses.
  Since 2SAT can be solved in polynomial time \cite{aspvall1979linear}, this yields a polynomial-time decision procedure.
\end{proof}

\section{Conclusions}
In this paper, we have developed algorithms for the synthesis
of reactive systems in distributed environments.  We have studied the
problem in the setting of Petri games.  Previously, the decidability of
Petri games was only known for non-distributed environments, i.e., for
games with a single environment token~\cite{petrigames}. Our
algorithms solve Petri games with one system token and an arbitrary number of environment
tokens.  We have shown that the synthesis
problem can be solved in polynomial time for nets with up to two environment tokens.  For an arbitrary but fixed
number of three or more environment tokens, the problem is
NP-complete. If the number of environment tokens grows with the size
of the net, the problem is EXPTIME-complete.

An intriguing question for future work is whether our results, which
scale to an arbitrary number of environment tokens, can be
combined with the results of~\cite{petrigames}, which scale to an arbitrary
number of system tokens. This would allow us to synthesize
``distributed systems in distributed environments.'' With the algorithm
presented in this paper, we can already synthesize individual components in
such distributed systems, by treating the other components as
adversarial (cf. \cite{DBLP:conf/atva/FinkbeinerS05}). 
The approach of \cite{petrigames} would additionally allow us to
analyze the cooperation between the system components.

\bibliography{bib}

\appendix

\section{Multisets}
A multiset $M$ over a set $S$ is a function from $S$ to the non-negative integers.
For an element $s$ of $S$, let $s \in M$ denote that $M(s) > 0$.
The set of all such elements $\{s \in S \mid M(s) > 0\}$ is called the support of $M$.
We identify $\{0,1\}$-valued multisets with their support.
A multiset $M$ over $S$ is finite if its support is finite.
The cardinality of such a finite multiset is defined as $|M| = \sum_{s \in S} M(s)$.
Otherwise, we write $|M| = \infty$.
For multisets $M, N$ over $S$, $M \subseteq N$ holds iff $M(s) \leq N(s)$ for all $s \in S$.
We define the difference of two multisets such that $(M - N)(s) = \max (0, M(s) - N(s))$ for all $s$.
Similarly, the disjoint union of multisets satisfies $(M + N)(s) = M(s) + N(s)$ for all $s$.
For a finite set $X = \{M_1, M_2, \dots, M_t\}$ of multisets, let $\biguplus_{M \in X} M$ denote $M_1 + M_2 \dots + M_t$.
If $f$ is a function from a set $S$ to a set $T$ and $M$ is a multiset over $S$, we define $f[M]$ to be the multiset over $T$ defined by $f[M](t) = \sum_{s \in S, f(s) = t} M(s)$.

\section{Detailed proofs}
\subsection{Relating cuts and markings}
\label{sec:cutsmarkings}

\begin{lemma}
  Let $\cut$ be a cut in an occurrence net $\net$, in which $t \in \transitions$ is enabled.
  Then, $\cut' \coloneqq \cut - \pre{t} + \post{t}$  is also a cut.
  In particular, $\cut'$ is a set and $\cut' = \cut \setminus \pre{t} \cup \post{t}$.
  \label{lem:cutstep}
\end{lemma}
  \begin{proof}
    Since $\cut$ is a set and $\post{t}$ is a set, $\cut'$ is a set if we can show that these two are disjoint.
    Let $p$ be any element of $\pre{t} \cap \cut$.
    If there were a element $p'$ of $\cut$ and $\post{t}$, it would hold that $p < p'$.
    This would contradict the fact that both are elements of a cut.
    Because of this, we can reason about sets and set operations in the following.

    It holds that every two distinct places $p, p'$ in $\cut'$ are concurrent:
    We distinguish different cases depending on whether the nodes lie in $\cut \setminus \pre{t}$ or in $\post{t}$ and depending on what could prevent them from being concurrent.
    \begin{itemize}
      \item If $p, p' \in \cut \setminus \pre{t}$, they are concurrent because $\cut$ is a cut.
      \item Assume $p,p' \in \post{t}$; $p \conflict p'$. Since places in occurrence nets have at most one incoming transition, $t$ would be in self-conflict, contradiction.
      \item Assume $p \in \post{t}, p' \in \cut \setminus \pre{t}$; $p \conflict p'$. Then, $t \conflict p'$ and there is $x \in \pre{t} \subseteq \cut$ such that $x < p'$ or $x \conflict p'$, both contradicting the assumption that $x$ and $p'$ are elements of a cut.
      \item Assume $p,p' \in \post{t}$; $p < p'$. Since $\pre{p'} = \{t\}$, $t \flows p < t \flows p'$. Thus, $t < t$ holds, which contradicts the well-foundedness of $\flow^{-1}$.
      \item Assume $p \in \post{t}, p' \in \cut \setminus \pre{t}$; $p < p'$. Let $x \in \pre{t} \subseteq \cut$. Then, $x < p < p'$, contradicting the assumption $x,p' \in \cut$.
      \item Assume $p \in \cut \setminus \pre{t}, p' \in \post{t}$; $p < p'$. Since $\pre{p'} = \{t\}$, $p < t$. Because $p \notin \pre{t}$, there is $x \in \pre{t}$ such that $p < x$, which contradicts the assumption that $x, p \in \cut$.
    \end{itemize}
    All other cases follow by symmetry.

    Furthermore, $\cut'$ is a maximal set of concurrent places because every place $p \notin \cut'$ concurrent to $\cut'$ would also be concurrent to $\cut$, but not be an element of it.
    This would contradict the maximality of the cut $\cut$.
  \end{proof}

\begin{lemma}
All reachable markings $\marking$ in an occurrence net $\net$ are cuts.
\label{lem:markingscuts}
\end{lemma}
\begin{proof}
By induction over the number of transitions needed to reach the marking from the initial one.
If $\marking$ is the initial marking, the proposition directly holds:
By definition, the initial marking is a set.
Because no place in it has an incoming transition, no two places can be causally related or in conflict.
Since following the inverse flow relation from any node will always lead us to an initial place, no other place is concurrent to $\initialmarking$.
Thus, $\initialmarking$ is a cut.
The induction step holds by Lemma~\ref{lem:cutstep}.
\end{proof}

\begin{lemma}
For every occurrence net $\net$ and every place $p \in \places$, $\past{p}$ is finite.
\label{lem:finitepast}
\end{lemma}
\begin{proof}
In an occurrence net, the indegree of every node is finite because all places have at most one incoming transition and because the size of the precondition of a transition is finite by definition.\footnote{Note that the postcondition of a place can be infinite. Thus, the same is not true about outdegrees.}
$\past{p}$ can be seen as the nodes of the reachable fragment of $\net$ with inverted flow starting from $p$ and, with this relation, the smaller graph has finite outdegree.
Furthermore, all paths in this graph are finite since $\flow^{-1}$ is well-founded.
Thus, by König's lemma, $\past{p}$ is finite as well.
\end{proof}

\begin{lemma}
  Let $\net$ be an occurrence net and $S$ be a finite set of pairwise concurrent places.
  Then, there is a reachable marking $\marking$ such that $S \subseteq \marking$.
  \label{lem:pcinmarking}
\end{lemma}
  \begin{proof}
    Prove by induction on $\left| \past{S} \setminus S \right|$, which is finite according to Lemma~\ref{lem:finitepast}.

    If $\left| \past{S} \setminus S \right| = 0$, then $S \subseteq \initialmarking$, which proves the claim.
    Else, $\pre{S} \neq \emptyset$.
    Choose $t \in \pre{S}$ such that there is no $t' \in \pre{S}$ with $t < t'$ and choose any $x \in S \cap \post{t}$.
    Choosing such a $t$ is possible because $\pre{S}$ is finite and $\flow^{-1}$ is well-founded.
    Then, $S' \coloneqq S \setminus \post{t} \cup \pre{t}$ is a finite set.
    All distinct $p,p'$ in $S'$ are concurrent:
    \begin{itemize}
      \item All $p,p' \in S \setminus \post{t}$ are concurrent by assumption.
      \item Assume $p,p' \in \pre{t}$; $p \conflict p'$. Then $t$ would be in self-conflict, contradiction.
      \item Assume $p \in \pre{t}, p' \in S \setminus \post{t}$; $p \conflict p'$. Then, $x \conflict p'$, contradiction.
      \item Assume $p,p' \in \pre{t}$; $p < p'$. If $p \flows t < p' \flows t$, $\flow^{-1}$ would not be well-founded. Otherwise, $t$ can be reached from $p$ both directly and via another transition, thus $t$ is in self-conflict, contradiction.
      \item Assume $p \in \pre{t}, p' \in S \setminus \post{t}$; $p < p'$. If $p \flows t < p'$, there would be a $t' \in \pre{S}$ such that $t < t'$, contradiction.
        Otherwise, $x \conflict p'$, contradiction.
      \item Assume $p \in S \setminus \post{t}, p' \in \pre{t}$; $p < p'$. $p < p' < x$, contradiction.
    \end{itemize}
    Furthermore, $\past{S'} \setminus S' \subsetneq \past{S} \setminus S$:
    First, $\past{S'} \setminus S' \subseteq \past{S'} \subseteq \past{S}$.
    Second, an element $s \in \past{S'} \setminus S'$ cannot be in $S'$ but also not in $\post{t}$, thus not in $S$.
    The set inclusion is strict because the right hand side includes $t$ while the left hand side does not.

    By the induction hypothesis, there is a reachable marking $\marking'$ such that $S' \subseteq \marking'$.
    In particular, $\pre{t} \subseteq S' \subseteq \marking'$.
    Then, $\marking' \markingstep{t} \marking$ such that $S \subseteq \marking$.
  \end{proof}

\begin{lemma}
  For every reachable marking $\marking^U$ in an initial branching process $(\net^U, \lambda)$ of a net $\net$, there is a reachable marking $\marking$ in $\net$ such that $\lambda[\marking^U] = \marking$.
  \label{lem:branchingmarkings}
\end{lemma}
  \begin{proof}
    By induction on the number of fired transitions needed to reach $\marking^U$ from $\initialmarking^U$.
    If $\marking^U$ is the initial marking, $\lambda[\marking^U] = \initialmarking$ because $\lambda$ is an initial homomorphism.

    For the induction step, assume that there are markings $\marking^U, \marking^{U \prime}$ and a transition $t \in \transitions^U$ such that $\marking^U \markingstep{t} \marking^{U \prime}$.
    Further assume that there is $\marking \in \reachablemarkings{\net}$ such that $\lambda[\marking^U] = \marking$.
    Then, $\lambda(t)$ is enabled in $\marking$ because $\pre{\lambda(t)} = \lambda[\pre{t}] \subseteq \lambda[\marking^U]$.
    Choose the reachable marking $\marking'$ such that $\marking \markingstep{\lambda(t)} \marking'$.
    \begin{align*}
      \lambda[\marking^{U \prime}] &= \lambda[\marking^U - \pre{t} + \post{t}]\\
      &= \lambda[\marking^U] - \lambda[\pre{t}] + \lambda[\post{t}] \\
      &= \marking - \pre{\lambda(t)} + \post{\lambda(t)} \\
      &= \marking'
      \qedhere % Move up the QED
    \end{align*}
  \end{proof}

\begin{lemma}
  In a branching process $(\net^U, \lambda)$ of a $k$-bounded, finite net $\net$, all cuts are finite.
  \label{lem:boundedfinitecuts}
\end{lemma}
  \begin{proof}
    By contradiction.
    Assume that there is an infinite cut $\cut$.
    By the pigeonhole principle, $k + 1$ distinct places in $\cut$ are mapped to the same value by $\lambda$.
    Since these places are pairwise concurrent, they are a subset of a reachable marking $\marking^U$ by Lemma~\ref{lem:pcinmarking}.
    By Lemma~\ref{lem:branchingmarkings}, there is a reachable marking $\marking$ in $\net$ such that $\lambda[\marking^U] = \marking$.
    This marking has at least $k + 1$ tokens in a single place, in contradiction to $k$-boundedness.
  \end{proof}

\begin{lemma}
  Let $\cut, \mathcal{D}$ be two finite cuts in an occurrence net $\net$ such that $\cut \subseteq \past{\mathcal{D}}$.
  Then, there exists a sequence $t_1, t_2, \dots, t_r$ of transitions such that $\cut \markingstep{t_1} \markingstep{t_2} \dots \markingstep{t_r} \mathcal{D}$.
  \label{lem:cutsreachable}
\end{lemma}
  \begin{proof}
    By induction on $n \coloneqq \left| \past{\mathcal{D}} \setminus \past{\cut} \right|$.
    This is finite by Lemma~\ref{lem:finitepast}.

    If $n=0$, $\past{\mathcal{D}} = \past{\cut}$ and thus $\mathcal{D} \subseteq \past{\cut}$.
    For every $c \in \cut$, there is $d \in \mathcal{D}$ such that $c \leq d$ and there is $c' \in \cut$ such that $d \leq c'$.
    Thus, $c \leq d \leq c'$ and, by the definition of a cut, $c = d = c'$.
    This means that $\cut$ is a subset of $\mathcal{D}$ and, by symmetry, $\cut = \mathcal{D}$.
    The claim then holds for the empty sequence.

    Now let $n > 0$ and therefore $\cut \neq \mathcal{D}$. Thus, there is at least one transition in $T \coloneqq \post{\cut} \cap \past{\mathcal{D}}$.
    Choose a transition $t \in T$ such that there is no transition $t' \in T$ with $t' < t$.
    This is possible since $\flow^{-1}$ is well-founded.
    Choose $x \in \pre{t} \cap \cut$.

    Next, we prove that $\pre{t} \subseteq \cut$ holds. It suffices to show that every $p \in \pre{t}$ is concurrent to each $c \in \cut$, $c \neq p$.
    Then, by the maximality of cuts, $p \in \cut$.
    \begin{itemize}
      \item Assume $p \conflict c$. Since $t \in \past{\mathcal{D}}$, there is a $d_t\in \mathcal{D}$ such that $t < d_t$.
        Moreover, there is $d_c \in \mathcal{D}$ such that $c \leq d_c$. It follows that $d_t \conflict d_c$ and, since both are elements of a cut, $d_t= d_c$.
        Then, the single incoming transition of $d_t$ is in self-conflict, which contradicts the assumption that $\net$ is an occurrence net.
      \item Assume $p < c$. If $t < c$, $x < t < c$ holds, which contradicts the fact that $x$ and $c$ are in the cut $\cut$.
        Else, there is a transition $t' \neq t$ such that, with $d_t$ and $d_c$ as above, $p \flows t' < c < d_c$ and furthermore $p \flows t < d_t$.
        Thus, $d_t \conflict d_c$, contradiction.
      \item Assume $c < p$. Let $t'$ be a transition such that $c \flows t' < p$. From $t' \in \post{\cut}$ and $t' < p \flows t < d_t$ follows $t' \in T$, which contradicts the choice of $t$.
    \end{itemize}
    Now let $\cut' \coloneqq \cut \setminus \pre{t} \cup \post{t}$, i.e., $\cut \markingstep{t} \cut'$.
    By Lemma~\ref{lem:cutstep}, $\cut'$ is a cut, $\cut'$ is finite and, by construction, $\cut' \subseteq \past{\mathcal{D}}$.
    It holds that $\past{\mathcal{D}} \setminus \past{\cut'} \subsetneq \past{\mathcal{D}} \setminus \past{\cut}$ because $\past{\cut} \subseteq \past{\cut'}$ and $t \in \past{\cut'} \setminus \past{\cut}$.
    By the induction hypothesis, there exists a sequence $t_2, \dots, t_r$ such that $\cut' \markingstep{t_2} \dots \markingstep{t_r} \mathcal{D}$, thus $\cut \markingstep{t} \markingstep{t_2} \dots \markingstep{t_r} \mathcal{D}$.
  \end{proof}

\begin{corollary}
  For a branching process of a bounded, finite net, the notion of a cut coincides with that of a reachable marking.
  \label{lem:boundedbranchingcutsmarkings}
\end{corollary}
\begin{proof}
  All reachable markings are cuts by Lemma~\ref{lem:markingscuts}.
  Every cut $\cut$ is finite by Lemma~\ref{lem:boundedfinitecuts}, $\initialmarking \subseteq \past{\cut}$ holds, and therefore $\cut$ is a reachable marking by Lemma~\ref{lem:cutsreachable}.
\end{proof}

\subsection{Theorem~\ref{thm:petrigraph}: Translating strategies from $\game$ to $\mathit{Graph}(\game)$}
\label{sec:petrigraphproof}
\petrigraph*
\begin{proof}
  Let $n$ be an arbitrary node of the tree, which is labeled $(\marking, c)$.
  As in the construction of the graph game strategy, we reach a cut $\cut$ such that $\lambda[\cut] = \marking$ by canonically firing transitions corresponding to the prefix.

  If $n$ were labeled with a bad vertex of type (X1), $\marking$ would be a bad marking.
  Then, the preimage of a bad marking is reachable in $\net^\sigma$, which contradicts the assumption that $\beta_\sigma$ ensures the safety condition.
  If the label of $n$ were a bad vertex of type (X2a), there would be distinct transitions $\innet{t_1},\innet{t_2} \in c$ such that $\pre{\innet{t_1}}, \pre{\innet{t_2}} \subseteq \marking$.
  As shown in the construction, there would exist a transition $t_1 \in \transitions^\sigma$ enabled in $\cut$ such that $\lambda(t_1) = \innet{t_1}$ and $s_\cut \in \pre{t_1}$.
  Symmetrically, there is a $t_2$ that is also enabled in $\cut$, has $s_\cut$ in its precondition and that is distinct from $t_1$ because $\lambda(t_2) = \innet{t_2} \neq \innet{t_1} = \lambda(t_1)$.
  Then, $\beta_\sigma$ would not be deterministic in $s_\cut$, contradiction.
  If the label of $n$ were a bad vertex of type (X2b), some enabled system transition $\innet{t}$ would have multiple preimages enabled in $\cut$ with different preimages of $\innet{p}$ in its precondition, which also implies nondeterminism.
  Finally, if $n$ were a bad vertex of type (X3), some $\innet{t} \in \transitions$ would be enabled in $\marking$, but all enabled transitions would involve the system and would not be in $c$.
  Thus, no purely environmental transitions would be enabled in $\cut$ either.
  Neither are there any system transitions enabled in $\cut$ because any transition $t$ in $\post{s_\cut}$ would mean that $\lambda(t) \in c$.
  If one such $t$ were enabled in $\cut$, $\lambda(t)$ would be enabled in $\marking$, which contradicts the assumption made above.
  Since there is no enabled transition in $\cut$ but $\innet{t}$ is enabled in $\marking$, $\beta_\sigma$ does not avoid deadlocks, contradiction.

  Thus, $n$ cannot be a bad vertex and $T_\sigma$ is a winning strategy.
\end{proof}

\subsection{Theorem~\ref{thm:graphgraphprime}: Translating strategies from $\mathit{Graph}(\game)$ to $\mathit{Graph'}(\game)$}
\label{sec:graphgraphprimeproof}
\graphgraphprime*
\begin{proof}
  Let $\Pi : \mathcal{V}' \to \mathcal{V}$ be the tuple projection onto the first two components.
  $\Pi$ is surjective, which allows us to think about $\mathcal{V}'$ as the vertices $\mathcal{V}$ with additional information.
  Furthermore, $\Pi(\mathcal{I}') = \mathcal{I}$, a vertex $v' \in \mathcal{V}'$ is a bad vertex iff $\Pi(v') \in \mathcal{X}$ and, whenever $v_1' \rightarrow v_2' \in \mathcal{E}'$ for $v_1',v_2' \in \mathcal{V}'$, $\Pi(v_1') \rightarrow \Pi(v_2') \in \mathcal{E}$ also holds.

  Denote a winning strategy for $\mathit{Graph}(\game)$ by $T_\sigma$.
  We inductively build a strategy $T_\sigma'$ for $\mathit{Graph'}(\game)$, choosing the successors of nodes labeled with vertices in $\mathcal{V}'_0$ as follows:
  Consider a node in $T_\sigma'$ such that the path leading to it from the root is labeled $v_0' \rightarrow v_1' \dots v_r'$ and that $(\marking, \top, \{\innet{s_\marking}\}) = v_r' \in \mathcal{V}_0'$.
  Assume that all previous choices by Player~0 have been constructed as we currently do.
  Then, $\Pi(v_0') \rightarrow \Pi(v_1') \dots \Pi(v_r')$ is the prefix of a play allowed by $T_\sigma$ due to the observations made above, and $T_\sigma$ gives us a successor edge $\Pi(v_r') \rightarrow (\marking, c)$.
  Correspondingly, $T_\sigma'$ chooses the outgoing edge $v_r' \rightarrow (\marking, c, \{\innet{s_\marking}\})$ of $v_r'$.

  Assume that $T_\sigma'$ would contain a node labeled by a bad vertex.
  Projecting the path from the root to this node under $\Pi$ gives us a play prefix allowed by $T_\sigma$, which also ends in a vertex in $\mathcal{X}$.
  Therefore, $T_\sigma$ would not be a winning strategy, which is a contradiction.
\end{proof}

\subsection{Lemma~\ref{lem:associatedlkc}: Responsibility multiset construction prevents illegal information flow}
\label{sec:associatedlkcproof}

\begin{lemma}
\label{lem:responsibility}
  Let a node in $T_\sigma$ be labeled $(\marking, c/\top, R)$ and have an associated cut $\cut$.
  Then, there is a subset $S \subseteq \cut$ such that $\lambda[S] = R$ and
  \begin{equation}
    \past{\cut} = \cut \cup \bigcup_{s \in S}{\past{s}}
  \label{eqresp}
  \end{equation}
\end{lemma}
\begin{proof}
  The inclusion ``$\supseteq$'' is clear for any $S \subseteq \cut$.
  We will thus only prove ``$\subseteq$'', by induction over the position of the node in $T_\sigma$.
  The claim clearly holds for the root as it is only mapped to the initial cut, for which $\past{\cut} = \cut$ holds.
  Assume now that the claim has already been established for the parent of our current node.
  We make a case distinction based on the type of the edge leading from the label of the parent to the current label:

  If this edge is of type (E'1), the proposition directly follows from the induction hypothesis because $\marking$, $R = \{\innet{s_\marking}\}$ and $\cut$ all stay the same.

  Else, the edge is of type (E'2) or (E'3).
  Let $\cut$ denote an associated cut of the parent from which $\cut'$, the cut of the current node, was formed.
  Moreover, let $S$ be the subset of $\cut$ given by the induction hypothesis and let the two vertices have the responsibility multisets $R$ and $R'$, respectively.
  Let $t$ be the transition added or used in this step of the construction, i.e., $\cut \markingstep{t} \cut'$.
  Choose $o \in \post{t}$ such that $\{\lambda(o)\} = R' - (R - \pre{\lambda(t)})$.
  We can always choose such an $o$, either as the preimage of $\innet{o}$ for an edge of type (E'2) or as the preimage of $\innet{s_{\marking'}}$ for an edge of type (E'3).
  Then, we will show that the claim holds for $S' \coloneqq S \setminus \pre{t} \cup \{o\}$, for which it clearly holds that $S' \subseteq \cut'$ and $\lambda[S'] = R'$.

  Let $n \in \past{\cut}$ be an arbitrary transition or place.
  It suffices to show that $n$ is an element of the expression on the right hand side of Eq.~\ref{eqresp}.
  By the construction of the strategy, $n$ must either be equal to $t$ or $n \in \post{t}$ or $n \in \past{\cut}$.
  If $n = t$ and more generally if $n \leq t$, $n$ is an element of the right hand side since it lies in $\past{o}$.
  Else, if $n \in \post{t}$, it lies in $\cut'$ and the claim holds.
    Thus, we can now assume $n \in \past{\cut}$ and, by the induction hypothesis, it holds that (\ref{lemrespc1}) $n \in \cut$ or (\ref{lemrespc2}) there is a place $s \in S$ such that $n \leq s$.
  \begin{enumerate}
    \item \label{lemrespc1} In the first case, the claim directly holds if $n \in \cut'$.
      Otherwise, $n \in \cut \setminus \cut' \subseteq \pre{t}$, so $n \leq t$ and the proposition follows as described above.
    \item \label{lemrespc2} In the second case, if $s \in \pre{t}$, we again have $n \leq t$ and are done.
      Else, $s \in S'$.
      \qedhere % Move up the QED
  \end{enumerate}%
\end{proof}

\begin{lemma}
Let $\net$ be an occurrence net, and $x \in \places$.
Then, $\mathit{LKC}(x)$ is a cut.
\end{lemma}
\begin{proof}
We show that all places in $\cut \coloneqq \mathit{LKC}(x)$ are pairwise concurrent and that $\cut$ is maximal with respect to this property.

Let $p$ and $p'$ be two distinct places in $\cut$.
From the definition of $\cut$, it follows that $p$ is the only element of $\past{p}$ with $p \nless x$.
Because $p' \nless x$ by assumption, $p'$ cannot lie in $p$'s causal past.
By symmetry, it follows that $p$ and $p'$ are not causally related.
Furthermore, $p$ and $p'$ are not in conflict:
If $\pre{p}=\emptyset$ or $\pre{p'}=\emptyset$, this is trivial because one of the places cannot be reached from any transition.
Otherwise, the single causal predecessors $t$ of $p$ and $t'$ of $p'$ would both be in $\past{x}$ and thus in $\past{t_x}$ for the predecessor $t_x$ of $x$.
If $p$ and $p'$ were in conflict, so would be $t$ and $t'$ and thus $t_x$ would be in self-conflict.
This is not possible because $\net$ is an occurrence net.
We conclude that all places in $\cut$ are pairwise concurrent.

To show maximality, consider any place $y \notin \cut$.
We will demonstrate that there is an element of $\cut$ that is not concurrent to $y$.
Why might $y$ not be in $\cut$?
It must hold that $y < x$ or that there is a single predecessor $t \in \mathit{pre}(y)$ such that $t \nless x$.
In the first case, $y$ is causally related to $x \in \cut$, and we are done.
In the second case, since we require nonempty preconditions, there must be a $y' \in \pre{t}$.
\begin{itemize}
  \item If $y' \in \cut$, we have shown that $y$ is causally related to a place in the cut and we are done.
  \item Else, it might be the case that $y' < x$.
    As $t \nless x$, there must be another transition $t' \in \post{y'}$ with $t' < x$.
    Then, $x$ and $y$ can be reached from $y'$ via two different outgoing transitions.
    Therefore, $x \in \cut$ and $y$ are in conflict and $y$ is not concurrent to all places in the cut.
  \item Finally, we might again have the situation that $y' \notin \cut$ because there is an incoming transition that is not in the causal past of $x$.
    Then, we repeat this step for an arbitrary place in the precondition of this transition.
    As we walk $\net$ in inverse flow direction, we build a sequence $y, y', y'', \dots$ of places.
    This process must terminate at some point because the inverse flow relation $\flow^{-1}$ is well-founded.
    Eventually, we will show either that an element in the past of $y$ is already in $\cut$ or that $x$ and $y$ are in conflict.
\end{itemize}
Thus, every place that does not lie in $\cut$ is not concurrent to all elements of the cut.
It follows that $\cut$ is a cut.
\end{proof}

\begin{lemma}
  \label{lem:lkcpast}
  Let $\cut$ be a cut in an occurrence net $\net$ and $c \in \cut$.
  Then, $\mathit{LKC}(c) \subseteq \past{\cut}$.
\end{lemma}
\begin{proof}
  For every $p \in \mathit{LKC}(c)$, we know that $p \nless c$ and that, if $p$ has an incoming transition $t$, $t < c$ holds.
  If $p$ were concurrent to all distinct elements of $\cut$, $p \in \cut$ by the maximality of cuts, and we are done.
  Otherwise, there must be $c' \in \cut$ such that $p$ and $c'$ are in conflict or causally related.
  If $p \conflict c'$, there is another place $x$ such that $p$ and $c'$ can be reached via different outgoing transitions from $x$.
  Furthermore, the single incoming transition $t$ of $p$ exists.
  Since $t < c$, it also holds that $c \conflict c'$, which contradicts the assumption that $c$ and $c'$ are elements of the cut $\cut$ and therefore concurrent.
  If $c' < p$, again by the assumption $t < c$, $c' < c$ would also hold, which would again contradict the fact that $c$ and $c'$ are concurrent.
  Thus, $p \leq c'$ holds, which shows the inclusion.
\end{proof}

\associatedlkc*
\begin{proof}
  We set $s \coloneqq s_\cut$.
  First, we show that $\cut \subseteq \text{LKC}(s)$.
  Because all elements $c$ of $\cut$ are concurrent, $c \nless s$.
  Let $t$ be the incoming transition of $c$ if it exists.
  By Lemma~\ref{lem:responsibility} and since the responsibility multiset of the node is $\{\lambda(s)\}$, $\past{\cut} = \cut \cup \past{s}$ holds.
  Because $t$ is a transition in $\past{\cut}$, $t < s$ follows.
  Thus, $c \in \text{LKC}(s)$ and the inclusion holds.  
  Due to the maximality of the cut $\cut$, the two cuts are equal.
\end{proof}

\subsection{Responsibility multiset construction not overly restrictive}
\label{sec:selectproof}

\begin{figure}[htp]
  \centering
  \begin{subfigure}[b]{.38\textwidth}
    \centering
    \includegraphics{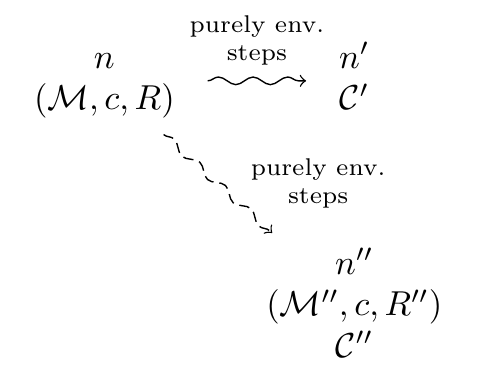}
    \caption{General situation.}
    \label{figlemselect}
  \end{subfigure} \quad
  \begin{subfigure}[b]{.58\textwidth}
    \centering
    \includegraphics{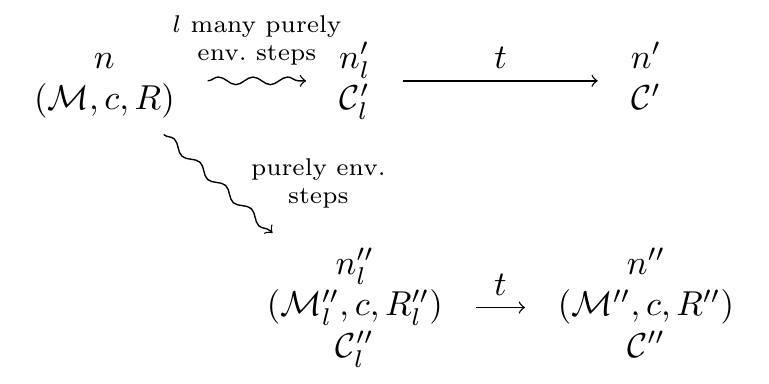}
    \caption{Induction step.}
    \label{figlemselectstep}
  \end{subfigure}
  \caption{Situation in Lemma~\ref{lemselect}. ``env.'' is short for ``environmental''.}
\end{figure}

\begin{lemma}
  \label{lemselect}
  Let $n$ be a node in $T_\sigma$ labeled with $(\marking, c, R)$ and let $n'$ be another node that is reachable from $n$ using only steps of type (E'2), i.e., steps corresponding to purely environmental transitions. Let $\cut'$ be a cut associated with $n'$ and let $S \subseteq \cut'$.
  Then, we can also reach a node $n''$, mapped to a cut $\cut''$ and labeled with $(\marking'', c, R'')$, from $n$ using only steps of type (E'2) such that $S \subseteq \cut''$ and $R'' \subseteq R + \lambda[S]$.
  
  (See Fig.~\ref{figlemselect} for a diagram of the situation.)
\end{lemma}
\begin{proof}
  By induction over the length $l$ of the path from $n$ to $n'$. If $l=0$, the claim trivially holds by choosing $n'' = n' = n$.
  If $l>0$, let $n_l'$ be the predecessor of $n'$ on the path from $n$ to $n'$.
  Call the associated cut that $\cut'$ was built from $\cut_l'$ and let $t \in \transitions^\sigma$ be the transition such that $\cut_l' \markingstep{t} \cut'$.
  This situation is shown in Fig.~\ref{figlemselectstep}.
  If $S \subseteq \cut_l'$, the conclusion directly follows from the induction hypothesis.
  Else, choose $b \in S \setminus \cut_l' \subseteq \cut' \setminus \cut_l' \subseteq \post{t}$.
  $S \setminus \post{t} \cup \pre{t}$ is a subset of $\cut_l'$ and therefore, by the induction hypothesis, a node $n_l''$ is reachable from $n$ with the label $(\marking_l'',c,R_l'')$ and an associated cut $\cut_l''$ such that $S \setminus \post{t} \cup \pre{t} \subseteq \cut_l''$ and $R_l'' \subseteq R + \lambda[S \setminus \post{t} \cup \pre{t}] \subseteq R + \lambda[S \setminus \post{t}] + \pre{\lambda(t)}$.
  We even have $S \setminus \post{t} \subseteq \cut_l'' \setminus \pre{t}$ because if there was a $b' \in S \cap \pre{t}$, $b' < t < b$ would hold, which contradicts the fact that $b$ and $b'$ are both members of the cut $\cut'$.
  Since the precondition of $t$ lies in $\cut_l''$, $n_l''$ must have a successor $n''$ with the label $(\marking'', c, R'')$ and the cut $\cut''$ such that $\cut_l'' \markingstep{t} \cut''$ and $R'' = R_l'' - \pre{\lambda(t)} + \{\lambda(b)\}$.
  Then, $S \subseteq \cut_l'' \setminus \pre{t} \cup \post{t} = \cut''$ and $R'' \subseteq R + \lambda[S \setminus \post{t}] + \pre{\lambda(t)} - \pre{\lambda(t)} + \{\lambda(b)\} \subseteq R + \lambda[S]$.
\end{proof}

\subsection{Theorem~\ref{thm:graphprimepetri}: Translating strategies from $\mathit{Graph'}(\game)$ to $\game$}
\label{sec:graphprimepetriproof}

\begin{lemma}[justified refusal for purely environmental transitions]
  Let $S \subseteq \envplaces^\sigma$ be a set of pairwise concurrent environment places such that $\lambda[S]=\pre{\innet{t}}$ for some $\innet{t} \in \transitions$.
  Then, there exists $t \in \transitions^\sigma$ such that $\pre{t}=S$ and $\lambda(t)=\innet{t}$.
  \label{lem:justenv}
\end{lemma}
\begin{proof}
  We add more concurrent places to $S$ to obtain a cut $\cut \supseteq S$.
  $s_\cut$ was added to the Petri game strategy for a node $n$ in the graph game strategy, which is labeled with a $\mathcal{V}_0'$ vertex and must have as one associated cut $\text{LKC}(s_\cut)$ by Lemma~\ref{lem:associatedlkc}.
  Since $\mathit{LKC}(s_\cut) \subseteq \past{\cut}$, there is a sequence of (purely environmental) transitions leading from $\mathit{LKC}(s_\cut)$ to $\cut$, by Lemma~\ref{lem:cutsreachable}.
  From the unique successor of $n$ in $T_\sigma$, a node with an associated cut $\cut$ is reachable using edges of type (E'2) corresponding to these transitions.
  From this node, there must be a successor step of type (E'2) that corresponds to firing $\innet{t}$ from the marking $\lambda[\cut]$.
  If no such transition did previously exist, a suitable $t$ was added when considering this node in the construction of the strategy.
\end{proof}

\begin{lemma}[justified refusal for system transitions]
  Let $s \in \sysplaces^\sigma$, $t \in \post{s}$ and $S \subseteq \places^\sigma$ be a set of pairwise concurrent places such that $s \in S$ and $\lambda[S] = \pre{\lambda(t)}$.
  Then, there exists a transition $t' \in \transitions^\sigma$ such that $\lambda(t')=\lambda(t)$ and $\pre{t'}=S$.
  \label{lem:justsys}
\end{lemma}
\begin{proof}
  Extend $S$ to a cut $\cut$.
  $t$ must have been inserted in a step of type (E'3) labeled $(\marking_t, c, R_t) \rightarrow (\marking_t', \top, \{\innet{s_t'}\})$ such that $\lambda(t) \in c$ and $s$ lies in one of the cuts associated with the first of the two nodes.
  We move towards the root in $T_\sigma$ from this node until we reach a $\mathcal{V}_0'$ node.
  Just like its single successor node $n$, it must have an associated cut $\mathit{LKC}(s)$.
  As above, from $n$, we can reach a node $n'$ with an associated cut $\cut$ using edges of type (E'2).
  By Lemma~\ref{lemselect}, we can reach a node $n''$, labeled $(\marking'', c, R'')$ and with an associated cut $\cut''$, from $n$ using edges of type (E'2) such that $S \subseteq \cut''$ and $R \subseteq \{\lambda(s)\} + \lambda[S]$.
  Since $\lambda[\cut'']$ cannot contain multiple system places and since $\lambda(s) \in \lambda[S]$, it even holds that $R \subseteq \lambda[S] = \pre{\lambda(t)}$.
  Furthermore, $\pre{\lambda(t)} = \lambda[S] \subseteq \lambda[\cut''] = \marking''$.
  Since $\lambda(t) \in c$, $n''$ must have a successor corresponding to firing a transition $\lambda(t)$.
  By the construction of the strategy, a transition $t'$ such that $\lambda(t') = \lambda(t)$ and $\pre{t'} = S$ must exist.
\end{proof}

\begin{lemma}[justified refusal]
  \label{lem:just}
  Let $S \subseteq \places^\sigma$ be a set of pairwise concurrent places and $\innet{t} \in \transitions$ such that $\lambda[S] = \pre{\innet{t}}$.
  Assume that there is no $t \in \transitions^\sigma$ that has $S$ as its precondition and satisfies $\lambda(t) = \innet{t}$.
  Then, there exists $s \in S \cap \sysplaces^\sigma$ such that none of the transitions in its postcondition is a $\lambda$-preimage of $\innet{t}$.
\end{lemma}
\begin{proof}
  By contraposition.
  Assume that for all system places $s$ in S, there exists a $t_s \in \post{s}$ such that $\lambda(t_s) = \innet{t}$.
  We need to show that there is $t \in \transitions^\sigma$ such that $S = \pre{t}$ and $\lambda(t) = \innet{t}$.
  If any such $s$ exists, the claim follows from Lemma~\ref{lem:justsys}.
  Else, $S \subseteq \envplaces^\sigma$ and the claim holds by Lemma~\ref{lem:justenv}.
\end{proof}

\begin{lemma}[determinism]
  \label{lem:determinism}
  Let $\cut$ be a cut in $\net^\sigma$ and $t_1, t_2 \in \transitions^\sigma$ such that $s_\cut \in \pre{t_i} \subseteq \cut$ for $i = 1,2$.
  Then, $t_1 = t_2$ holds.
\end{lemma}
\begin{proof}
  Since both transitions have a system place in their precondition, each $t_i$ must have been inserted for a step corresponding to an edge of type (E'3) in the graph game strategy.
  This step goes out from a node $n_i$ labeled with a vertex $(\marking_i, c_i, R_i)$ such that $\lambda(t_i) \in c_i$.
  $s$ lies in one of the cuts associated with $n_i$.
  Call the latest $\mathcal{V}_0'$ node on the path from the root $m_i$ and call its label $(\marking_i', \top, \{\lambda(s)\})$.
  These nodes must have an associated cut containing $s$ as well.
  According to Lemma~\ref{lem:associatedlkc}, $m_1$ and $m_2$ are both associated with the cut $\mathit{LKC}(s)$ and thus $\marking_1' = \lambda[\mathit{LKC}(s)] = \marking_2'$.
  Since $T_\sigma$ is memoryless, the unique successors of the nodes $m_i$ must be labeled with the same vertex $( \marking_1', c, \{ \lambda(s) \})$.
  Because all steps between the respective successor and $n_i$ correspond to edges of type (E'2), the commitment $c$ is preserved and $c_1 = c = c_2$.
  By Lemma~\ref{lem:lkcpast}, $\mathit{LKC}(s) \subseteq \past{\cut}$ holds, and thus, by Lemma~\ref{lem:cutsreachable}, we can reach $\cut$ from $\mathit{LKC}(s)$ by firing a sequence of transitions.
  All these transitions must be purely environmental because $s$ is preserved in the sequence of cuts.
  From $m_1$ onwards, we follow the corresponding sequence of nodes and finally reach a node labeled with $\left( \lambda[\cut], c, R\right)$, mapped to $\cut$.
  Then, both $\lambda(t_1)$ and $\lambda(t_2)$ lie in $c$ and are enabled in $\lambda[\cut]$.
  If $\lambda(t_1)$ and $\lambda(t_2)$ were distinct, this node would be labeled with a losing vertex of type (X'2a), which contradicts the assumption that $T_\sigma$ is a winning strategy.
  Thus, $\lambda(t_1) = \lambda(t_2)$.
  If $\pre{t_1} \neq \pre{t_2}$, there would be a place $p \in \pre{t_1} - \pre{t_2}$.
  Then, $0 < \pre{\lambda(t)}(\lambda(p)) < \lambda[\cut](\lambda(p))$.
  Therefore, the node would be labeled with a losing vertex of type (X'2b), contradiction.
  Thus, $\lambda(t_1) = \lambda(t_2)$ and $\pre{t_1} = \pre{t_2}$ and, because $\beta_\sigma$ is a branching process, $t_1 = t_2$.
\end{proof}

\begin{lemma}[deadlock avoidance]
  \label{lem:deadlock}
  Let $\cut$ be a cut in $\net^\sigma$ and $\innet{t} \in \transitions$ such that $\pre{\innet{t}} \subseteq \lambda[\cut]$. Then, there exists $t' \in \transitions^\sigma$ such that $\pre{t'} \subseteq \cut$.
  (Note that we do not require $\lambda(t') = \innet{t}$.)
\end{lemma}
\begin{proof}
  If $\innet{t}$ is purely environmental, this directly follows from Lemma~\ref{lem:justenv}.
  We can thus assume that no purely environmental transition is enabled in $\lambda[\cut]$.
  $s_\cut$ was added to the strategy for a node labeled with a vertex in $\mathcal{V}_0'$, which has a unique successor $n$ labeled $(\lambda[\mathit{LKC}(s_\cut)], c, \{\lambda(s_\cut)\})$.
  Both have as an associated cut $\text{LKC}(s_\cut)$ by Lemma~\ref{lem:associatedlkc}.
  From $n$, we can reach a node that is mapped to the cut $\cut$ and is labeled with $\left( \lambda[\cut], c, R \right)$, only using edges of type (E'2).
  Since this node cannot be labeled with a bad vertex of type~(X'3) and since $\innet{t}$ is enabled in the marking $\lambda[\cut]$, but no purely environmental transition is enabled by assumption, there has to be a transition $\innet{t'} \in \transitions$ such that $\lambda(s_\cut) \in \pre{\innet{t'}} \subseteq \lambda[\cut]$ and $\innet{t'} \in c$.
  By choosing $S \subseteq \cut$ such that $\lambda[S] = \pre{\innet{t'}}$ and by applying Lemma~\ref{lemselect}, we can reach another node $n'$ from $n$ that is labeled $(\marking', c, R')$ with an associated cut $\cut'$ such that $S \subseteq \cut'$ and $R' \subseteq \{\lambda(s_\cut)\} + \lambda[S]$.
  As argued in Lemma~\ref{lem:justsys}, it even holds that $R' \subseteq \lambda[S]$.
  $\innet{t'}$ is enabled in $\lambda[\cut'] = \marking'$, and $R' \subseteq \pre{\innet{t'}}$.
  From $n'$, the environment can make a move corresponding to the execution of $\innet{t'}$.
  Thus, a preimage $t'$ of $\innet{t'}$ must have been inserted into $\transitions^\sigma$ at the latest when visiting this step, and $\pre{t'} \subseteq \cut$ holds as claimed.
\end{proof}

\subsection{Theorem~\ref{thm:exphardness}: EXPTIME-hardness}
\label{sec:exphardnessproof}

\begin{figure}[tb]
  \centering
  \begin{subfigure}[b]{.38\textwidth}
    \centering
    \includegraphics[width=\textwidth]{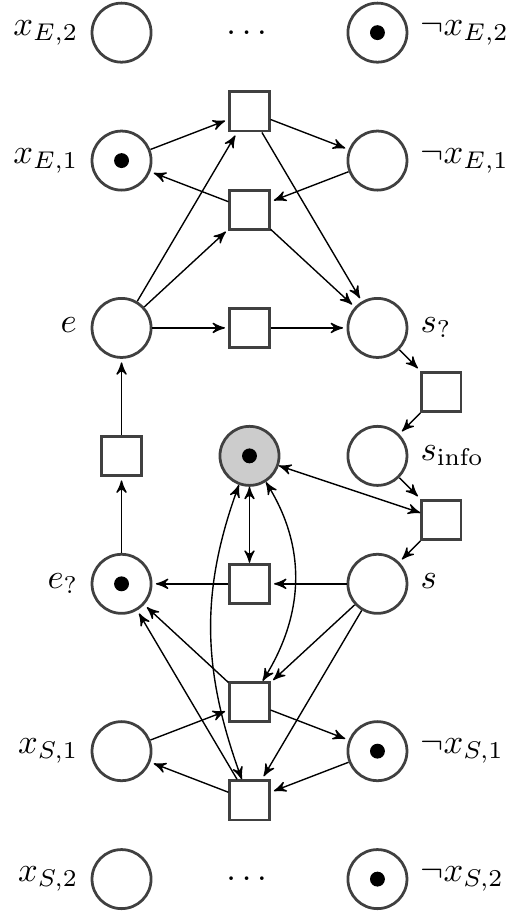}
    \caption{First half of the game.}
    \label{fig:1vnhardness1}
  \end{subfigure}
  \quad
  \begin{subfigure}[b]{.57\textwidth}
    \centering
    \includegraphics[width=\textwidth]{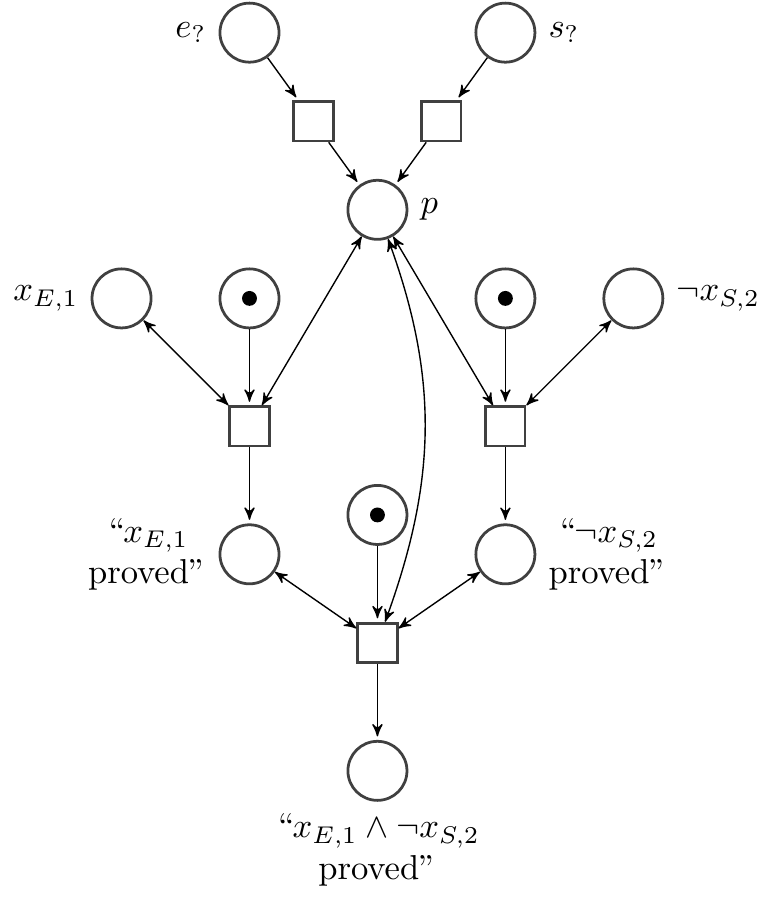}
    \caption{Second half.}
    \label{fig:1vnhardness2}
  \end{subfigure}
  \caption{Example of the EXPTIME-hardness reduction from $G_5$.}
  \label{fig:1vnhardness}
\end{figure}

\exphardness*
\begin{proof}
  The combinatorial game $G_5$ has already been introduced in the proof sketch.

  Conceptually, the constructed game can be divided into two parts, illustrated in Fig.~\ref{fig:1vnhardness1} and \ref{fig:1vnhardness2}.
  The most important token is an environment token, which we will call the turn counter.
  In the first half of the game, it cycles between five places: $e_?$, $e$, $s_?$, $s_\text{info}$ and $s$.
  If $P_S$ is the first to move, this token starts in $s_?$; otherwise, it starts in $e_?$.
  The place $s$ ($e$) signals that the next fired transition corresponds to a move by $P_S$ ($P_E$).
  $s_?$ (or $e_?$) means that the environment can choose to advance to $s$ (or $e$) or can instead end the first part of the game and prove that $\phi$ currently holds.
  Finally, $s_\text{info}$ is an intermediary place between $s_?$ and $s$.
  If the turn counter is there, the environment has already chosen to let $P_S$ make her next move, but must still inform her of the last move made by her opponent.

  Each boolean variable is represented by an environment token located in one of two places, representing its possible truth values.
  In the beginning, the positions of these tokens correspond to the initial configuration of the variables.
  If it is $P_E$'s turn, i.e., if the turn counter lies in $e$, the environment can toggle the place of one variable in $X_E$ and move the turn counter to $s_?$ simultaneously using a joint transition or simply move the turn counter in order to pass.
  The single system player is located in one single system place and is responsible for choosing the moves made by $P_S$.
  She gets informed about previous moves via a transition leading from $s_\text{info}$ to $s$ and from the system place back to itself.
  Then, the system player chooses her move by allowing one of the transitions, each of which leave the system player in her current place, move the turn counter from $s$ to $e_?$ and toggle one variable in $X_S$.
  Again, there is one more transition that does not toggle any variable, which is used for passing.

  Any time the turn counter is in $s_?$ or $e_?$, the environment can choose to allow the next move by going to $s_\text{info}$ or $e$.
  Alternatively, the environment can choose to prove that the formula is now satisfied and thereby win the game.
  In order to achieve this, the turn counter token can permanently leave the places mentioned so far and go to a fresh place $p$.
  This starts the second half of the game.
  From then on, the variables do not change any more because every such transition requires the turn counter to lie in either $e$ or $s$.
  For every subformula of $\phi$, we introduce one more environment token with two places.
  All such tokens start out in a place that represents that the environment has not yet proved the respective subformula to be true.
  The turn counter can now move them to the second place to prove the subformula.
  Literals can only be proved if the corresponding variable had the right value.
  In order to prove conjunctions, the counter token takes a joint transition that moves the subformula token and additionally synchronizes with the tokens corresponding to both conjuncts, which must have been proved before.
  Disjunctions can be proved with one of two transitions, which only require synchronization with one of the disjuncts, which must have been previously proved.
  The place corresponding to having proved the entire formula is marked as a bad place, i.e., the bad markings are exactly the markings containing this place.

  This means that the environment can win iff $P_E$ wins the $G_5$ game by satisfying the formula.
  The system player is forced to keep playing the game because she would otherwise cause a deadlock.
  She can base her decisions on the current state of all variables.
  For this reason, winning strategies between the two games can be translated in both directions.
\end{proof}

\subsection{Theorem~\ref{thm:nphardness}: NP-hardness for $p \geq 3$ environment players}
\label{sec:nphardnessproof}
\nphardness*
\begin{proof}
The upper bound has been shown at the beginning of Section~\ref{sec:sparse}.

We show hardness by a reduction from the canonical NP-complete problem 3SAT.
Assume that we are given a formula $\phi$ over the variables $X=x_1,\dots,x_r$ of the following shape:
\[ \bigwedge_{1 \leq i \leq n}{\underbrace{C_{i,1} \lor C_{i,2} \lor C_{i,3}}_{C_i}}\text{,} \]
where all $C_{i,j}$ are literals over the variables in $X$.

We begin constructing our Petri game by adding initial places
for the system and the three environment players.
For every clause $C_i$, we add a joint transition for the
three environment tokens to three fresh places.
  We will later refer to these transitions as \emph{clause transitions}.
Furthermore, we label the places that the transitions lead to
with the literals of the clause.
It is important that we add distinct nodes for different
appearances of the same literal. From each
of these places, we offer a transition together with the
system token at its initial place to sink places for
both tokens.
After doing this for all clauses, we pick all pairs of places
labeled with complementary literals, i.e., where one literal
is the negation of the other.
For each such pair, we add a new transition from the three
initial environment places to the pair of places as well
as one additional state with a self-loop transition.
Because these transitions will be used by the environment
to find a contradiction in an invalid assignment, we call them
\emph{contradiction transitions}.
The game does not have any bad markings.

\begin{figure}
  \centering
  \includegraphics[width=.8\textwidth]{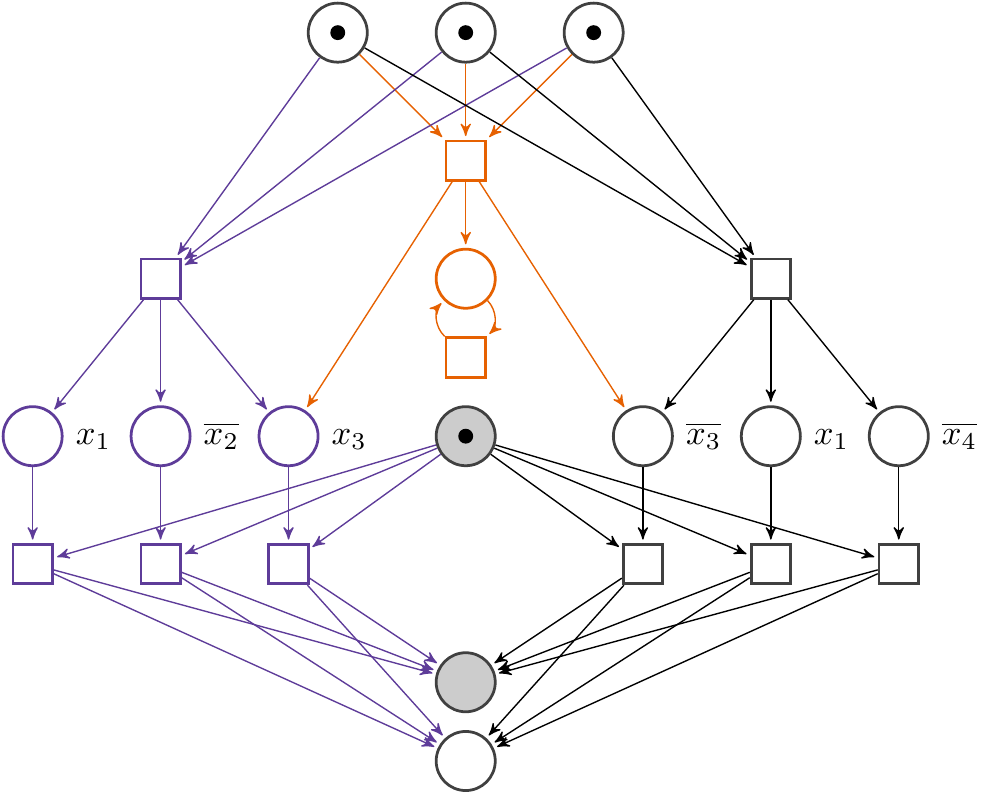}

  \caption{Example reduction for the formula $(x_1 \lor \neg x_2 \lor x_3) \land (\neg x_3 \lor x_1 \lor \neg x_4)$.}
  \label{fig:3satreduc}
\end{figure}

Figure~\ref{fig:3satreduc} gives an example of the reduction for the
SAT instance $(x_1 \lor \neg x_2 \lor x_3) \land (\neg x_3 \lor x_1 \lor \neg x_4)$.
  The parts of the game that are added for the first of the two clauses are marked in violet, while all parts added for the pair of complementary literals $x_3$ and $\neg x_3$ are colored in orange.

  It remains to prove that, for every instance $\phi$ of 3SAT, the constructed game has a strategy iff $\phi$ was satisfiable:

  \begin{description}
  \item[``$\Leftarrow$''] Assume that $\phi$ has a satisfying assignment.
  Consider an arbitrary clause of the formula.
  In this clause, the assignment satisfies at least one literal.
  In the construction of the game, we added one place per literal and from each of these places, there is one joint transition with the system.
  The Petri game strategy will allow exactly one of these transitions, which must correspond to an arbitrarily chosen literal satisfied in the assignment, and forbid the other two--even if they are satisfied as well.
  Doing this for every clause uniquely characterizes a strategy.

  While the strategy would trivially satisfy the safety condition because there are no bad markings at all, it is not clear that this strategy is deterministic and deadlock avoiding.
  We will show that this is the case regardless of how the environment acts:
  Assume that the environment takes a clause transition and thus the three environment tokens are in the three places that correspond to the three literals of the clause.
  The only enabled transitions are the three transitions for the literal places.
  Because our strategy allows exactly one of the three, there is neither a deadlock nor nondeterminism in this marking.\footnote{This is also the reason why we need distinct places for the same literal occurring in different clauses. Otherwise, we might be forced to allow multiple transitions corresponding to literals in the same clause. This then would cause nondeterminism in the initial system place.}
  After taking the chosen transition, the system cannot move anymore.
  Thus, the system does not actively provoke a deadlock and there is no nondeterminism involved.
  Assume now that the environment chooses a contradiction transition instead.
  Then, two environment tokens lie in places corresponding to complementary literals and one token lies in the place with the self-loop transition.
  Because this last token can move infinitely often by itself, deadlock is no longer possible.
  There now are two enabled transitions in which the system can participate.
  Because we only allowed transitions whose literals are satisfied in the assignment, they cannot be both allowed at the same time and thus the strategy is deterministic in the system place.

  \item[``$\Rightarrow$''] Assume now that there is a Petri game strategy for the system.
  From this, we will iteratively build a set of literals, which can be completed to form a satisfying assignment for $\phi$.
  For every clause, consider the case where the environment took the corresponding clause transition.
  Three transitions are enabled and all involve the system.
  Because the strategy is deadlock avoiding, at least one of the transitions belonging to the literals must be allowed.
  We add all such literals to our set and continue with the next clause.
  The resulting set contains at least one literal of every clause.
  Furthermore, it contains no pair of complementary literals:
  If the contrary were the case, the strategy would allow two transitions belonging to two complementary literals.
  Then, the environment could take the matching contradiction transition and thus enable both of these transitions at the same time.
  Both transitions involve the initial system place and would be allowed by the strategy.
  Then, the strategy would not be deterministic in this place, which is a contradiction.
  We can thus conclude that there is an assignment of boolean values that satisfies all literals in the collected set.
  Every such assignment satisfies every single clause and thus the whole formula.
  This proves that $\phi$ is satisfiable. \qedhere
  \end{description}
\end{proof}

\subsection{Theorem~\ref{thm:polynomialtwo}: Polynomial algorithm for $p \leq 2$ environment players}
\label{sec:polynomialtwoproof}
\polynomialtwo*
\begin{proof}
In order to get a polynomial-time algorithm, it suffices to replace solving the SAT instances generated in the algorithm presented in Appendix~\ref{sec:algodec} by polynomial-time computations.
For this, we optimize the generation of the SAT instance to generate only clauses with at most two literals.
Then, the resulting 2SAT instance can be decided and a solution be constructed in linear time as described in \cite{aspvall1979linear}.

Assume that we are currently at a vertex $(\marking, \top)$ and that we are looking for a commitment such that Player~1 cannot reach a bad vertex or a vertex in $\mathit{Attr}_i$ without steps by her opponent.
In a preprocessing step, we check whether a bad vertex is reachable, and if so, no adequate commitment exists and we return.
In another preprocessing step, we collect all system transitions that can be used by Player~1 to reach $\mathit{Attr}_i$.
Since they cannot end up in the commitment anyway, we eliminate these transitions for the further analysis.
Similarly, we remove all transitions that cannot end up in the commitment because they would lead to bad vertices of type~(X2b).
This corresponds to all constraints $\neg \innet{t}$ added in line~\ref{x2bend} of Algorithm~\ref{alg}.
Without loss of generality, we then only need to consider the reachability of bad vertices of types (X2a) and (X3).
In other words, we only consider deadlock constraints and nondeterminism constraints for distinct transitions (as opposed to nondeterminism of the same transition firing from different preconditions in the marking).

First, we observe that we do not have to distinguish between multiple transitions with the same precondition.
Choosing two transitions with the same precondition is unnecessary:
Because they are enabled in exactly the same markings, either one would suffice to avoid deadlock and choosing only one of them can only reduce nondeterminism, not increase it.
Moreover, if choosing one transition satisfies deadlock avoidance and determinism, any other transition with the same precondition would do so as well.
Thus, we can optimize the algorithm by considering only one arbitrary transition per precondition.
None of the other transitions will be added to the commitment.

If any system transition enabled in $\marking$ has no environment places in its precondition, we can commit to just this transition in order to satisfy determinism and deadlock avoidance.
Otherwise, we consider all markings reachable via purely environmental transitions.
Call such a marking $\marking' = \{\innet{s_\marking}, \innet{e_1}, \innet{e_2}\}$.
In $\marking'$, there are only three possible preconditions for enabled system transitions: $\{\innet{s_\marking}, \innet{e_1}\}$, $\{\innet{s_\marking}, \innet{e_2}\}$ or $\{\innet{s_\marking}, \innet{e_1}, \innet{e_2}\}$.
As described above, we only need to consider one transition per precondition.
If they exist, we call them $\innet{t_1}$, $\innet{t_2}$ and $\innet{t_{1,2}}$, respectively.

We cannot directly add the constraints as in the main algorithm since deadlock could still add 3-clauses of the shape $\innet{t_1} \lor \innet{t_2} \lor \innet{t_{1,2}}$ to our formula.
This only happens in the case where all three types of transitions are available and where no purely environmental transitions are enabled (else, no deadlock constraint is added).
However, this case is easy to deal with:
Since $\innet{t_{1,2}}$ is only enabled in this single marking, enabling or disabling it does not influence deadlock and nondeterminism in other reachable markings.
Because of this, we do not use $\innet{t_{1,2}}$ in the SAT formula and do not add any constraint for deadlock in $\marking'$.
For nondeterminism, we only add the constraint $\neg \innet{t_1} \lor \neg \innet{t_2}$.
If the final SAT instance is unsatisfiable, it would not be satisfiable with the original, stronger constraints either.
Else, we select $\innet{t_{1,2}}$ for the commitment iff both $\innet{t_1}$ and $\innet{t_2}$ are disabled in the satisfying assignment.
If one of them is selected already, the deadlock constraint is satisfied and disabling $\innet{t_{1,2}}$ ensures determinism.
Else, adding $\innet{t_{1,2}}$ to the commitment avoids deadlock in $\marking'$ and cannot lead to nondeterminism since no other transition in the commitment is enabled in $\marking'$.

These steps generate a 2SAT formula, which can be solved in linear time.
We then either report the existence of such a strategy or complete the satisfying assignment by the left out three-player transitions as described above to get an adequate commitment for building the strategy.
\end{proof}

\section{Evaluating commitments symbolically in the graph game}
\label{sec:algorithm}

\subsection{Description}
\label{sec:algodec}
\begin{algorithm}
  \DontPrintSemicolon
  \KwData{A bounded Petri game $\game$ with one system player.}
  \KwResult{If $\game$ has a strategy, return $\top$. Else, return $\bot$.}
  $\mathit{Attr}_0 \coloneqq \emptyset$\;
  \For(\tcp*[f]{$\leq \left| \reachablemarkings{\net} \right| + 1$ iterations}){$i \coloneqq 0$ \KwTo $\infty$}{
    $\mathit{Attr}_{i+1} \coloneqq \mathit{Attr}_i$\;
    \ForAll{$\marking \in \reachablemarkings{\net} \setminus \mathit{Attr}_i$}{
      $\mathit{sat} \coloneqq \top$\;
      $\mathit{reachable} \coloneqq $ markings reachable from $\marking$ via purely environmental transitions\; \nllabel{reachablesearch}
      \ForAll{$\marking' \in \mathit{reachable}$}{ \nllabel{reachableforbegin}
        \If{$\marking' \in \badmarkings$}{
          $\mathit{sat} \coloneqq \bot$\;
          \textbf{break}\;
        }
        \ForAll{$\innet{t}, \innet{t'} \in \transitions$}{ \nllabel{x2begin}
          \If{$\innet{t} \neq \innet{t'} \land \text{both are system transitions and enabled in $\marking'$}$}{ \nllabel{x2mid}
            $\mathit{sat} \coloneqq \mathit{sat} \land (\neg \innet{t} \lor \neg \innet{t'})$ \nllabel{x2end} \;
          }
        }
        \ForAll{$\innet{t} \in \transitions$}{ \nllabel{x2bbegin}
          \If{$\innet{t} \text{ enabled in } \marking' \land 0 < \pre{\innet{t}}(\innet{p}) < \marking'(\innet{p}) \text{ for some } \innet{p} \in \places$}{ \nllabel{x2bmid}
            $\mathit{sat} \coloneqq \mathit{sat} \land \neg \innet{t}$ \nllabel{x2bend} \;
          }
        }
        $\mathit{enabled} \coloneqq \text{enabled transitions in $\marking'$}$\;
        \If{$\mathit{enabled} \neq \emptyset \land \mathit{enabled}$ only contains system transitions}{
          $\mathit{sat} \coloneqq \mathit{sat} \land \bigvee_{\innet{t} \in \mathit{enabled}} \innet{t}$\;
        }
        \ForAll{$\innet{t} \in \transitions$}{
          \If{$\innet{t} \text{ enabled system transition in $\marking'$} \land \marking' \markingstep{\innet{t}} \marking'' \land \marking'' \in \mathit{Attr}_i$}{
            $\mathit{sat} \coloneqq \mathit{sat} \land \neg \innet{t}$ \nllabel{reachableforend} \;
          }
        }
      }
      \If{$\mathit{sat}$ is unsatisfiable}{
        $\mathit{Attr}_{i+1} \coloneqq \mathit{Attr}_{i+1} \cup \{\marking\}$\;
        \If{$\marking = \initialmarking$}{
          \textbf{return} $\bot$\;
        }
      }
    }
    \If{$\left| \mathit{Attr}_{i+1} \right| = \left| \mathit{Attr}_i \right|$}{
      \textbf{return} $\top$\;\nllabel{posreturn}
    }
  }
  \caption{Deciding Petri games}
  \label{alg}
\end{algorithm}

To speed up solving $\mathit{Graph}(\game)$ in practice, we evaluate commitments symbolically.
Pseudocode for the algorithm is presented in Algorithm~\ref{alg}.

The algorithm computes an increasing sequence $\mathit{Attr}_0 \subsetneq \mathit{Attr}_1 \subsetneq \mathit{Attr}_2 \dots$ of subsets of $\mathcal{V}_0$ until some $\mathit{Attr}_i$ contains the initial vertex $\mathcal{I}$ or until a fixed point is reached.
Because the second component of every vertex is $\top$, the pseudocode represents every vertex $(\marking, \top)$ by $\marking$, whereas we will speak about sets of vertices to stress the connection to the graph game.
For every $i$, $\mathit{Attr}_i$ contains exactly the vertices of Player~0 from which Player~1 can enforce to reach a bad vertex within at most $i$ steps by Player~0.

We initialize $\mathit{Attr}_0$ with the empty set since $\mathcal{X} \subseteq \mathcal{V}_1$ and thus no bad marking can be reached without a move by Player~0.

In every subsequent step, we set $\mathit{Attr}_{i+1}$ to contain all vertices in $\mathit{Attr}_i$ and additionally all vertices $(\marking, \top)$ such that, for all commitments $c$, Player~0 can reach a bad vertex or a vertex in $\mathit{Attr}_i$ from $(\marking, c)$ using only transitions of types (E2) and (E3).
To find out whether $(\marking, \top)$ should be added, we explore all markings reachable from $\marking$ via purely environmental transitions.
For every such reachable marking $\marking'$, we add constraints on the commitment $c$ expressing that $(\marking', c)$ is not a bad vertex and that $(\marking', c)$ does not have a type-(E3) successor in $\mathit{Attr}_i$.
These constraints can be expressed as propositional formulas in conjunctive normal form over the variables $\innet{t} \in \post{s_\marking}$ where variable $\innet{t}$ represents the proposition that $\innet{t} \in c$:

\begin{itemize}
  \item $(\marking', c)$ is a bad vertex of type (X1) iff $\marking' \in \badmarkings$, regardless of the commitment. If such a vertex is reachable, no commitment can prevent Player~1 from winning. Because of this, we add the constraint $\bot$ and stop searching for additional clauses, which would not change the satisfiability of the formula.
  \item $(\marking', c)$ is a bad vertex of type (X2a) iff there are two different transitions $\innet{t_1}, \innet{t_2} \in c$ that are both enabled in $\marking'$. Because of this, we add a conjunct of the following form, where $\innet{t}$ and $\innet{t'}$ range over those transitions in $\post{\innet{s_\marking}}$ that are enabled in $\marking'$:
    \[ \bigwedge_{\substack{\innet{t}, \innet{t'} \\ \innet{t} \neq \innet{t'}}}{\neg \innet{t} \lor \neg \innet{t'}} \]
  \item $(\marking', c)$ is a bad vertex of type (X2b) iff there is an enabled transition $\innet{t} \in c$ and a place $\innet{p}$ such that $0 < \pre{\innet{t}}(\innet{p}) < \marking(\innet{p})$.
    If this is the case, we add a conjunct of the shape $\neg \innet{t}$.
  \item $(\marking', c)$ is a bad vertex of type (X3) iff only system transitions $\innet{t} \notin c$ are enabled in $\marking'$ and at least one such transition is enabled. If any purely environmental transition is enabled in $\marking'$ or if no transition is enabled at all, do not add any constraints. Otherwise, we require one enabled transition to be in $c$, which we express by $\bigvee_{\innet{t}}{\innet{t}}$ where $\innet{t}$ ranges over the enabled transitions in $\marking'$.
  \item $(\marking', c)$ has a successor in $\mathit{Attr}_i \subseteq \mathcal{V}_0$ iff, for any $(\marking_S, \top) \in \mathit{Attr}_i$, there is a $\innet{t} \in c$ such that $\marking' \markingstep{\innet{t}} \marking_S$. Thus, we add the conjunct $\neg \innet{t}$ for any such $\innet{t}$.
\end{itemize}

We then use a SAT solver to decide the satisfiability of the conjunction of all these constraints.
If there is a satisfying assignment, it describes a commitment that can avoid both directly losing and entering $\mathit{Attr}_i$.
If no such assignment exists, Player~1 can force her opponent into a bad vertex within at most $i + 1$ steps by Player~0, which is why we add the vertex to $\mathit{Attr}_{i+1}$.

If, at any point, the initial vertex is added to some $\mathit{Attr}_i$, there is no winning strategy for $\mathit{Graph}(\game)$ and thus no strategy for $\game$.
Else, we iteratively compute the described step until a fixed point $\mathit{Attr}_i = \mathit{Attr}_{i+1}$ is reached.
Since the sets grow in any previous step, this will be the case after at most $\left| \reachablemarkings{\net} \right| + 1$ iterations.
For every vertex outside the fixed point, choose the outgoing edge corresponding to the commitment given by the satisfying assignment of the SAT formula.
This describes a memoryless graph game strategy that never enters a bad vertex and is thus winning.

By executing the reductions in Section~\ref{sec:graphpetri}, we obtain a strategy for $\game$.
This means that, from a system place $s$, the system allows exactly those outgoing transitions whose labels are elements of the commitment chosen from $\lambda[\mathit{LKC}(s)]$.
This uniquely describes a strategy.

\subsection{Asymptotic runtime}
\label{sec:algocomp}
We will now give an upper bound on the runtime of the algorithm.
We want to give this upper bound as a function of three parameters of a game $\game$, namely the number of reachable markings $r$, the number of transitions $t$ and the maximum number of players $p \coloneqq \max{\left\{|\marking| \,\middle|\, \marking \in \reachablemarkings{\net} \right\}}$.

Our runtime analysis does not capture the time for parsing the description of the game, since it might grow arbitrarily large for fixed parameters $r,t,p$.
In particular, the number of bad markings can technically have a great influence on the runtime since it can grow exponentially large in the size of the underlying net.
If the bad markings are given individually, just reading them in might take a long time even though the actual number of bad markings plays no role in the algorithm itself.
In practice, one might want to use a more concise encoding to reduce parsing time.
We simply assume that it is possible to check whether a given marking is bad in $\mathcal{O}(t^2\,p)$ since the total running time will then be dominated by other computations.
For example, recognizing a bad marking can be achieved in just $\mathcal{O}(p)$ using a hash set for individual markings or using a hash set for individual places if the bad markings are specified through a set of bad places.
The time for parsing the input and building suitable data structures must be added to our time bound and might be important, for example if the number of bad markings is much larger than the number of reachable markings or if there are large numbers of unreachable places.
For reasonable inputs however, this term will be vastly dominated by the runtime of the algorithm.
In addition, we will (falsely) assume that the identifiers of transitions and places can be read, written and manipulated in constant time, ignoring the logarithmic factors accounting for the increasing size of identifiers.

The body of the loop between line~\ref{reachableforbegin} and \ref{reachableforend} is dominated by the generation of clauses for bad vertices of type~(X2a) in lines~\ref{x2begin} to \ref{x2end}, taking time $\mathcal{O}(t^2\,p)$.
This leads to a runtime of $\mathcal{O}(r\,t^2\,p)$ for the loop as a whole.
This term dominates the graph search in line~\ref{reachablesearch}, which only takes $\mathcal{O}(r\,t\,p)$ time.\footnote{The implicit graph that is explored has the $r$ many markings as its nodes and every node has $\leq t$ outgoing edges, corresponding to the enabled transitions in a marking. Checking whether a transition is enabled and computing the successor marking after firing a transition can be done in $\mathcal{O}(p)$ if markings, preconditions and postconditions are represented as sorted lists of places.}

The size of each generated SAT formula is bounded by a function in $\mathcal{O}(r\,t^2)$ and the number of its variables by $t$.
A naïve enumeration of assignments and checking whether the formula holds for any of them takes $\mathcal{O}(r\,2^t\,t^2)$.
We expect an off-the-shelf SAT solver to give much better performance for typical games.

The number of iterations of both outer loops is linearly bounded in $r$.
Even though the outermost loop has no explicit exit condition, the size of $\reachablemarkings{\net} \setminus \mathit{Attr}_i$ decreases in every iteration until the algorithm returns in line~\ref{posreturn}.
In the worst case, we add one marking per iteration and detect the fixed point only in the $(r + 1)$th iteration.
By taking into account the two outer loops, we obtain a runtime in $\mathcal{O}(r^2\,(r\,t^2\,p + r\,2^t\,t^2)) = \mathcal{O}(r^3\,t^2(2^t + p))$.

Since the number of reachable markings is bounded by $(k + 1)^{|\places|}$ and since the maximum number of players is bounded by $k\cdot|\places|$, the algorithm is still exponential in the size of the underlying net.

\end{document}